%% file: main_LIP_fullversion.tex
\documentclass[letterpaper,USenglish]{lipics-v2018}

\usepackage{framed}
\usepackage[framemethod=tikz]{mdframed}
\usepackage{wrapfig}
\usepackage{cleveref}
\usepackage{microtype}
\usepackage{boxedminipage}
\allowdisplaybreaks
\usetikzlibrary{shapes.geometric}
\usetikzlibrary{arrows}
\usetikzlibrary{arrows.meta}
\usetikzlibrary{patterns}
\usetikzlibrary{shapes.misc}

\newcommand{\local}{${\mathsf{LOCAL}}$}
\newcommand{\congest}{${\mathsf{CONGEST}}$}

\newcommand{\dist}{\mbox{\rm dist}}
\newcommand{\Hcal}{\mathcal{H}}

\newcommand{\polylog}{\mathsf{polylog}} 
\newcommand{\NearestNeighbors}{\mathsf{NearestNeighbors}}
\newcommand{\ConsSpannerSparseRegion}{\mathsf{SpannerSparseRegion}}
\newcommand{\ConsSpannerDenseRegion}{\mathsf{SpannerDenseRegion}}

\newcommand{\LocalSpanner}{\mathsf{LocalSpanner}}

\newcommand{\pr}[1]{\Pr\! \left[ {#1} \right]}
\newcommand{\prob}[2]{\Pr_{#1}\! \left[ #2 \right]}
\newcommand{\E}[1]{\mathop{\mathbf E}\! \left[ {#1} \right]}
\newcommand{\BB}{\{ 0,1 \}}
\newcommand{\PRG}{\mathsf{PRG}}
\newcommand{\Exp}[2]{\mathop{\mathbf E}_{#1}\! \left[ {#2} \right]}
\newcommand{\ith}[1]{{#1}\textsuperscript{th}}
\newcommand{\ie}  {i.e.,\ }
\newcommand{\eg}  {e.g.,\ }
\newcommand{\set}[1]{\left\{#1\right\}}

\newcommand{\cH}{{\mathcal H}}

\def\cA{{\cal A}}
\def\cC{{\cal C}}

\def\cA{{\cal A}}

\def\cC{{\cal C}}

\renewcommand{\paragraph}[1]{\vspace{0.15cm}\noindent {\bf #1}}
\newtheorem{observation}[theorem]{Observation}
\newtheorem{fact}[theorem]{Fact}
\newtheorem{claim}[theorem]{Claim}
\def\Proof{\par\noindent{\bf Proof:~}}
\def\blackslug{\hbox{\hskip 1pt \vrule width 4pt height 8pt
    depth 1.5pt \hskip 1pt}}
\def\QED{\quad\blackslug\lower 8.5pt\null\par}

\title{Congested Clique Algorithms for Graph Spanners}
\author{Merav Parter}{Weizmann IS, Rehovot, Israel}{merav.parter@weizmann.ac.il}{}{}
\author{Eylon Yogev}{Weizmann IS, Rehovot, Israel}{eylon.yogev@weizmann.ac.il}{}{}
\authorrunning{M. Parter and E. Yogev}
\Copyright{Merav Parter and Eylon Yogev}
\subjclass{Theory of computation, Distributed Algorithms}

\nolinenumbers
\begin{document}
\maketitle

\begin{abstract}
Graph spanners are sparse subgraphs that faithfully preserve the distances in the original graph up to small stretch.
Spanner have been studied extensively as they have a wide range of applications ranging from distance oracles,
labeling schemes and routing to solving linear systems and spectral sparsification.
A $k$-spanner maintains pairwise distances up to multiplicative factor of $k$.
It is a folklore that for every $n$-vertex graph $G$, one can construct a $(2k-1)$ spanner with $O(n^{1+1/k})$ edges. In a distributed setting, such spanners can be constructed in the standard \congest\ model using $O(k^2)$ rounds, when randomization is allowed.  

In this work, we consider spanner constructions in the congested clique model, and show:
\begin{itemize}
\item
A randomized construction of a $(2k-1)$-spanner with $\widetilde{O}(n^{1+1/k})$ edges in $O(\log k)$ rounds. The previous best algorithm runs in $O(k)$ rounds.
\item
A deterministic construction of a $(2k-1)$-spanner with $\widetilde{O}(n^{1+1/k})$ edges in $O(\log k +(\log\log n)^3)$ 
rounds. The previous best algorithm runs in $O(k\log n)$ rounds. This improvement is achieved by a new derandomization theorem for hitting sets which might be of independent interest.
\item
A deterministic construction of a $O(k)$-spanner with $O(k \cdot n^{1+1/k})$ edges in $O(\log k)$ rounds.
\end{itemize}
\end{abstract}


\input{intro}
\input{nearestn-full}
\input{sparse-full}
\input{densespanner}

\input{hitting-set-full}
\input{alter-det-full}

\vspace{-12pt}
\section*{Acknowledgments}
\vspace{-6pt}
The first author is grateful for Mohsen Ghaffari for 
earlier discussions on randomized spanner constructions in the congested clique 
via dynamic streaming ideas. The second author thanks Roei Tell 
for pointing out 
\cite{GopalanMRTV12}.

\bibliographystyle{plainurl}
\bibliography{detspanner} 

\appendix


\input{appendix}

\end{document}

%% file: intro.tex
\section{Introduction \& Related Work}

Graph spanners introduced by Peleg and Sch{\"a}ffer \cite{peleg1989graph} are 
fundamental graph structures, more precisely, subgraphs of an input graph $G$, 
that faithfully preserve the distances in $G$ up to small multiplicative 
stretch. 
Spanners have a wide-range of distributed  applications \cite{Peleg:2000} for routing \cite{thorup2001compact}, broadcasting, synchronizers \cite{peleg1989optimal}, and shortest-path computations \cite{BeckerKKL17}. 

The common objective in distributed computation of spanners is to achieve the best-known existential size-stretch trade-off within small number of \emph{rounds}. It is a folklore that for every graph $G=(V, E)$, there exists a $(2k-1)$-spanner $H \subseteq G$ with $O(n^{1+1/k})$ edges. Moreover, this size-stretch tradeoff is believed to be optimal, by the girth conjecture of Erd\H{o}s.

There are plentiful of distributed constructions of spanners for both the \local\ 
and the \congest\ models of distributed computing 
\cite{derbel2006fast,baswana2007simple,derbel2007deterministic,DerbelGPV08,derbel2009local,pettie2010distributed,derbel2010sublinear,grossmanparter17}.
 The standard setting is a synchronous message passing model where per round 
each node can send one message to each of its neighbors. In the \local\ model, 
the message size is unbounded, while in the \congest\ model it is limited to 
$O(\log n)$ bits. One of the most notable distributed randomized constructions 
of $(2k-1)$ spanners is by Baswana \& Sen \cite{baswana07} which can be 
implemented in $O(k^2)$ rounds in the \congest\ model.

Currently, there is an interesting gap between deterministic and randomized 
constructions in the \congest\ model, or alternatively between the 
deterministic construction of spanners in the \local\ vs.\ the \congest\ model. 
Whereas the deterministic round complexity of $(2k-1)$ spanners in the \local\ 
model is $O(k)$ due to \cite{DerbelGPV08}, the best deterministic algorithm in 
the \congest\ model takes $O(2^{\sqrt{\log n\cdot\log\log n}})$ rounds 
\cite{GFDetSpanner18}.

We consider the 
\emph{congested clique} model, introduced by Lotker et al. 
\cite{lotker2003mst}. In this model, in every round, each vertex can send 
$O(\log n)$ bits to each of the vertices in the graph.
The congested clique model has been receiving a lot of attention recently due to its relevance to overlay networks and large scale distributed computation  \cite{hegeman2015lessons,ghaffari2018improved,behnezhad2018brief}. 

\paragraph{Deterministic local computation in the congested clique model.} 
Censor et al. 
\cite{Censor-HillelPS16} initiated the study of \emph{deterministic} local 
algorithms in the congested clique model by means of derandomization of randomized \local\ algorithms. The approach of \cite{Censor-HillelPS16} can be summarized as follows. The randomized 
complexity of  the classical local problems is $\polylog(n)$ rounds
(in both \local\ and \congest\ models). For these randomized algorithms, it is 
usually sufficient 
that the random choices made by vertices are sampled from distributions with 
\emph{bounded independence}.  
Hence, any round of a randomized algorithm can be simulated by giving 
all nodes a shared random seed of $\polylog(n)$ bits.

To completely derandomize 
such a round, 
nodes should compute (deterministically) a 
seed which is at least as ``good''\footnote{The random seed is usually shown 
provide a large progress in expectation. The deterministically computed seed 
should provide a progress at least as large as the expected progress of a 
random seed.} as a random seed would be. This is achieved by
estimating their ``local progress'' when simulating the random choices using 
that 
seed. Combining the techniques of conditional expectation, pessimistic estimators and 
bounded independence, leads to a simple ``voting''-like algorithm in which the 
bits of the seed are computed 
\emph{bit-by-bit}. The power of the congested clique is hence in providing some global leader that collects all votes in $1$ round and broadcasts the winning bit value. This approach led to deterministic MIS in $O(\log \Delta \log n)$ rounds and deterministic $(2k-1)$ spanners with $\widetilde{O}(n^{1+1/k})$ edges in $O(k \log n)$ rounds, which also works for weighted graphs. Barenboim and Khazanov \cite{Barenboim18} presented deterministic local algorithms as a function of the graph's \emph{arboricity}.

\paragraph{Deterministic spanners via derandomization of hitting sets.}
As observed by 
\cite{roditty2005deterministic,bhattacharyya2012transitive,GFDetSpanner18}, the 
derandomization of the Baswana-Sen algorithm boils down into a derandomization 
of $p$-dominating sets or \emph{hitting-sets}. It is a well known fact that 
given 
a collection of $m$ sets $\mathcal{S}$, each containing at least $\Delta$ 
elements coming from a universe of size $n$, one can construct a hitting set 
$Z$ of size $O((n \log m)/\Delta)$. A randomized construction of such a set is 
immediate by picking each element into $Z$ with probability $p$ and applying 
Chernoff. A centralized deterministic construction is also well known by the 
greedy approach (\eg Lemma 2.7 of \cite{bhattacharyya2012transitive}).

In our setting we are interested in deterministic constructions of hitting sets 
in the congested clique model. In this setting, each vertex $v$ knows a subset 
$S_v$ of size at least $\Delta$, that consists of vertices in the $O(k)$-neighborhood of $v$, and it is required to compute a small set $Z$ 
that hits (i.e., intersects) all subsets.  Censor et al. \cite{Censor-HillelPS16} showed that the 
above mentioned randomized construction of hitting sets still holds with 
$g=O(\log n)$-wise independence, and presented an $O(g)$-round algorithm that 
computes a hitting set deterministically by finding a good seed of $O(g \log 
n)$ bits. Applying this hitting-set algorithm to compute each of the $k$ levels 
of clustering of the Baswana-Sen algorithm resulted in a deterministic $(2k-1)$ 
spanner construction with $O(k \log n)$ rounds. 
\vspace{-10pt}
\subsection*{Our Results and Approach in a Nutshell}
We provide improved randomized and deterministic constructions of graph 
spanners in the congested clique model. Our randomized solution is based on an 
$O(\log k)$-round algorithm that computes the $O(\sqrt{n})$ nearest vertices in 
radius $k/2$ for every vertex $v$\footnote{To be more precise, the algorithm computes the $O(n^{1/2-1/k})$ nearest vertices at distance at most $k/2-1$.}. This induces a partitioning of the graph 
into sparse and dense regions. The sparse region is solved ``locally'' and the 
dense region simulates only two phases of Baswana-Sen, leading
to a total round complexity of $O(\log k)$. We show the following for $n$-vertex unweighted graphs.
\begin{mdframed}[hidealllines=true,backgroundcolor=gray!25]
\vspace{-8pt}
\begin{theorem}\label{lem:randspanner}
There exists a randomized algorithm in the congested clique model that constructs a $(2k-1)$-spanner with $\widetilde{O}(k \cdot n^{1+1/k})$ edges within $O(\log k)$ rounds w.h.p.
\vspace{-3pt}
\end{theorem}
\end{mdframed}

Our deterministic algorithms are based on constructions of 
hitting-sets with short seeds. Using the pseudorandom generator of Gopalan et 
al.\ \cite{GopalanMRTV12}, we construct a hitting set with seed length 
$O(\log n \cdot (\log\log n)^3)$ which yields the following for $n$-vertex unweighted graphs.

\begin{mdframed}[hidealllines=true,backgroundcolor=gray!25]
\vspace{-8pt}
\begin{theorem}\label{lem:detspanner}
There exists a \emph{deterministic} algorithm in the congested clique model that constructs a $(2k-1)$-spanner with $\widetilde{O}(k\cdot n^{1+1/k})$ edges within $O(\log k+(\log\log n)^3)$ rounds.
\vspace{-3pt}
\end{theorem}
\end{mdframed}
In addition, we also show that if one settles for stretch of $O(k)$, then a hitting-set seed of $O(\log n)$ bits is sufficient for this purpose, yielding the following construction:
\begin{mdframed}[hidealllines=true,backgroundcolor=gray!25]
\vspace{-8pt}
\begin{theorem}\label{lem:detspanner2}
There exists a \emph{deterministic} algorithm in the congested clique model that constructs a $O(k)$-spanner with $O(k \cdot n^{1+1/k})$ edges within $O(\log k)$ rounds.
\vspace{-3pt}
\end{theorem}
\end{mdframed}

A summary of our results\footnote{\cite{BaswanaS07} does not mention the congested clique model, but the best randomized solution in the congested clique is given by simulating \cite{BaswanaS07}.} are given in the Table \ref{table-results}. All 
results in 
the table are with respect to spanners with $\widetilde{O}(n^{1+1/k})$ edges 
for an unweighted $n$-vertex graph $G$.
\input{table}
In what follows we provide some technical background and then present the high level ideas of these construction. 

\paragraph{A brief exposition of Baswana-Sen \cite{baswana07}.}
The algorithm is based on constructing $k$ levels of clustering $\mathcal{C}_0,\ldots, \mathcal{C}_{k-1}$, where a
clustering $\mathcal{C}_i=\{C_{i,1},\ldots,\}$ consists of  vertex disjoint subsets which we call clusters. Every cluster $C \in \mathcal{C}_i$ has a special node that we call \emph{cluster center}. For each $C \in \mathcal{C}_i$, the spanner contains a depth-$i$ tree rooted at its center and spanning all cluster vertices. 
Starting with the trivial clustering $\mathcal{C}_0=\{\{v\}, v \in V\}$, in 
each phase $i$, the algorithm is given a clustering $\mathcal{C}_i$ and it 
computes a clustering $\mathcal{C}_{i+1}$ by sampling the cluster center of 
each cluster in $\mathcal{C}_{i-1}$ with probability $n^{-1/k}$. Vertices that 
are adjacent to the sampled clusters join them and the remaining vertices 
become unclustered. For the latter, the algorithm adds some of their edges to 
the spanner. This construction yields a $(2k-1)$ spanner with $O(kn^{1+1/k})$ 
edges in expectation.

It is easy to see that this algorithm can be simulated in the congested clique 
model using $O(k)$ rounds. 
As observed in \cite{roditty2005deterministic,grossmanparter17}, the only 
randomized step in Baswana-Sen is picking the cluster centers of the 
$(i+1)^{th}$ clustering. That is, given the $n^{1-i/k}$ cluster centers of 
$\mathcal{C}_{i}$, it is required to compute a subsample of $n^{1-(i+1)/k}$ 
clusters without having to add to many edges to the spanner (due to unclustered vertices). This is exactly the hitting-set problem where the neighboring clusters of each vertex 
are the sets to cover, and the universe is the set of centers in 
$\mathcal{C}_{i}$ (ideas along these lines also appear in 
\cite{roditty2005deterministic,GFDetSpanner18}). 

\paragraph{Our Approach.} 
In the following, we provide the high level description of our construction while omitting many careful details and technicalities. We note that some of these technicalities stems from the fact that we insist on achieving the (nearly) optimal spanners, as commonly done in this area. Settling for an $O(k)$-spanner with $\widetilde{O}(k n^{1+1/k})$ edges could considerably simplify the algorithm and its analysis.
The high-level idea is simple and it is based on dividing the graph $G$ into 
sparse edges and dense edges,  constructing a spanner for each of these 
subgraphs using two different techniques. This is based on the following 
intuition inspired by the Baswana-Sen algorithm.

In Baswana-Sen, the vertices that are clustered in level-$i$ of the clustering 
are morally vertices whose $i$-neighborhoods is sufficiently \emph{dense}, 
i.e., containing at least $n^{i/k}$ vertices. We then divide the vertices into 
\emph{dense} vertices $V_{dense}$ and \emph{sparse} vertices $V_{sparse}$, 
where $V_{dense}$ consists of vetices that have $\Omega(\sqrt{n})$ vertices in 
their $k/2$-ball, and $V_{sparse}$ consists of the remaining vertices. This 
induces a partition of $G$ edges into $E_{sparse}=(V_{sparse} \times V)\cap 
E(G)$ and $E_{dense}$ that contains the remaining $G$-edges, i.e., edges whose 
both endpoints are dense. 

\paragraph{Collecting Topology of Closed Neighborhood.}
One of the key-building blocks of our construction is an $O(\log k)$-round 
algorithm that computes for each vertex $u$ the subgrpah $G_{k/2}(u)$ induced 
on its closest $O(\sqrt{n})$ vertices within distance at most $k/2$ in $G$. 
Hence the algorithm computes the entire $k/2$-neighborhoods for the sparse 
vertices. For the sake of the following discussion, assume that the maximum 
degree in $G$ is $O(\sqrt{n})$. Our algorithm handles the general case as well. 
Intuitively, collecting the $k/2$-neighborhood can be done in $O(\log k)$ 
rounds if the graph is sufficiently \emph{sparse} by employing the graph 
exponentiation idea of \cite{LenzenW10}. In this approach, in each phase the 
radius of the collected neighborhood is doubled. Employing this technique in 
our setting gives raise to several issues. First, the input graph $G$ is not 
entirely sparse but rather 
consists of interleaving sparse and dense regions, i.e.,  the
$k/2$-neighborhood of a sparse vertex might contain dense vertices. For that 
purpose, in phase $i$ of our algorithm, each vertex (either sparse or dense) 
should obtain a subset of its closest $O(\sqrt{n})$ vertices in its $2^{i}$ 
neighborhood. Limiting the amount collected information is important for being 
able to route this information via Lenzen's algorithm \cite{lenzen2013route} 
in $O(1)$ rounds in each phase. 

Another technicality concerns the fact that the relation ``$u$ is 
in the $\sqrt{n}$ nearest vertices to $v$'' is not necessarily symmetric. This 
entitles a problem where a given vertex $u$ is ``close''\footnote{By 
\emph{close} we mean being among the $\sqrt{n}$ nearest vertices.} to many 
vertices $w$, and $u$ is not close to any of these vertices. In case where 
these $w$ vertices need to receive the information from $u$ regarding its 
closest neighbors (i.e., where some their close vertices are close to $u$), $u$ 
ends up sending too many messages in a single phase. To overcome this, we 
carefully set the growth of the radius of the collected neighborhood in the 
graph exponentiation algorithm. We let only vertices that are close to 
each other to exchange their topology information and show that this is sufficient for computing the $G_{k/2}(u)$ subgraphs.  This procedure is the basis 
for or constructions as explained 
next.

\paragraph{Handling the Sparse Region.}
The idea is to let every sparse vertex $u$ locally simulate a \local\ spanner algorithm on its subgraph $G_{k/2}(u)$. For that purpose, we show that the deterministic spanner algorithm of \cite{DerbelGPV08} which takes $k$ rounds in general, in fact requires only $k/2$ rounds when running by a sparse vertex $u$. This implies that the subgraph $G_{k/2}(u)$ contains all the information needed for $u$ to locally simulate the spanner algorithm. 
This seemingly harmless approach has a subtle defect.
Letting only the sparse vertices locally simulate a spanner algorithm might 
lead to a case where a certain edge $(u,v)$ is not added by a sparse vertex due 
to a decision made by a dense vertex $w$ in the local simulation $u$ in 
$G_{k/2}(u)$. Since $w$ is a dense vertex it did not run the algorithm locally 
and hence is not aware of adding these edges. To overcome that, the sparse 
vertices notify the dense vertices about their edges added in their local 
simulations. We show how to do it in $O(1)$ rounds. 

\paragraph{Handling the Dense Region.}
In the following, we settle for stretch of $(2k+1)$ for ease of description. 
By applying the topology collecting procedure, every dense vertex $v$ obtains a set 
$N_{k/2}(v)$ consisting of its closest $\Theta(\sqrt{n})$ vertices within distance $k/2$. 
The main benefit in computing these $N_{k/2}(v)$ sets, is that it allows the dense vertices to ``skip'' over the first 
$k/2-1$ phases of Baswana-Sen, ready to apply the $(k/2)$ phase.

As 
described earlier, picking the centers of the clusters can be done by computing 
a hitting set for the set $\mathcal{S}=\{N_{k/2}(v), 
~\mid~ v \in V_{dense}\}$. It is easy to construct a random subset $Z \subseteq 
V$ of cardinality $O(n^{1/2})$ that hits all these sets and to cluster all the
dense vertices around this $Z$ set. This creates clusters of strong diameter $k$ (in the spanner) that cover all the dense vertices. The final step connects each pair of adjacent clusters by adding to the spanner a single edge between each such pair, this adds $|Z|^2=O(n)$ edges to the spanner.

\paragraph{Hitting Sets with Short Seed.}
The description above used a randomized solution to the following hitting set problem: given $n$ subsets of vertices $S_1,\ldots,S_n$, each $|S_i|\geq \Delta$, find a small set $Z$ that intersects all $S_i$ sets.
A simple randomized 
solution is to choose each node $v$ to be in $Z$ with probability $p=O(\log n 
/\Delta)$. 
The standard approach for derandomization is by using distributions with 
limited independence. Indeed, for the randomized solution to hold, it is sufficient
to sample the elements from a $\log n$-wise distribution. However, sampling an 
element with probability $p=O(\log n 
/\Delta)$ requires roughly $\log n$ random 
bits, leading to a total seed length of $(\log^2 n)$, which is too large for our 
purposes.

Our key observation is that for any set $S_i$ the event that $S_i \cap Z \ne 
\emptyset$ can be expressed by a {\em read-once DNF formula}. Thus, in order to 
get a short seed it suffices to have a pseudoranom generator (PRG) that can 
``fool'' read-once DNFs. A PRG is a function that gets a short random seed 
and expands it to a long one which is indistinguishable from a random seed of the same length for such 
a formula. Luckily, such PRGs with seed length of $O(\log n \cdot (\log\log n)^3)$ exist due to Gopalan et 
al.\ \cite{GopalanMRTV12}, leading to deterministic hitting-set algorithm with $O((\log\log n)^3)$ 
rounds.

\paragraph{Graph Notations.}
For a vertex $v \in V(G)$, a subgraph $G'$ and an integer $\ell \in \{1,\ldots, n\}$, let $\Gamma_{\ell}(v,G')=\{ u ~\mid~ \dist(u,v,G')\leq \ell\}$. When $\ell=1$, we omit it and simply write $\Gamma(v,G')$, also when the subgraph $G'$ is clear from the context, we omit it and write $\Gamma_{\ell}(v)$. 
For a subset $V' \subseteq V$, let $G[V']$ be the 
induced subgraph of $G$ on $V'$.  Given a disjoint subset of vertices $C,C'$, let $E(C,C',G)=\{(u,v) \in E(G) ~\mid~ u \in C \mbox{~~and~~} v\in C\}$.
we say that $C$ and $C'$ are \emph{adjacent} if $E(C,C',G)\neq \emptyset$. Also, for $v \in V$, $E(v,C,G)=\{(u,v) \in E(G) ~\mid~ u \in C\}$. 
A vertex $u$ is \emph{incident} to a subset $C$, if $E(v,C,G)\neq \emptyset$.

\paragraph{Road-Map.} \Cref{sec:topology} presents algorithm $\NearestNeighbors$ to collect the topology of nearby vertices. At the end of this section, using this collected topology, the graph is partitioned into sparse and dense subgraphs. \Cref{sec:detsanner} describes the spanner construction for the sparse regime. \Cref{dense} considers the dense regime and is organized as follows. First, \Cref{sec:denseg} describes a deterministic construction spanner given an hitting-set algorithm as a black box. Then, \Cref{sec:hitting} fills in this missing piece and shows deterministic constructions of small hitting-sets via derandomization. Finally, \Cref{sec:det2} provides an alternative deterministic construction, with improved runtime but larger stretch.

%% file: table.tex
\begin{flushleft}
	\centering
	\begin{tabular}{|c|c|c|c|}
		\hline
									& \textbf{Stretch} 	& 
									\textbf{\#Rounds}             		& 
									\textbf{Type}                           \\ 
									\hline
		{Baswana \& Sen \cite{BaswanaS07}}		& $2k-1$    & 
		$O(k)$                    	
		& \multirow{2}{*}{Randomized}    \\ 
		This Work 					& $2k-1$    & $O(\log k)$             	
		&                                \\ \hline
		{Censor-Hillel et al.~\cite{Censor-HillelPS16}}	& $2k-1$    & $O(k\log 
		n)$            	& 
		\multirow{3}{*}{Deterministic} \\ 
		This Work 					& $2k-1$    & $O(\log k + (\log \log 
		n)^3)$ 
		&                                \\ 
		This Work 					& $O(k)$ 	& $O(\log k)$                 	
		&                                \\ \hline
	\end{tabular}
	\label{table-results}
\end{flushleft}

%% file: nearestn-full.tex
\vspace{-15pt}
\section{Collecting Topology of Nearby Neighborhood}\label{sec:topology}
For simplicity of presentation, assume that $k$ is even,  for 
$k$ odd, we replace the term $(k/2-1)$ with $\lfloor k/2 \rfloor$. 
In addition, we assume $k\geq 6$. Note that 
randomized constructions with $O(k)$ rounds are known and hence one benefits 
from an $O(\log k)$ algorithm for a non-constant $k$. In the full version, we show 
the improved deterministic constructions for $k \in \{2,3,4,5\}$.
\vspace{-10pt}
\subsection{Computing Nearest Vertices in the $(k/2-1)$ Neighborhoods}\label{sec:nearsestn}
In this subsection, we present an algorithm that computes the $n^{1/2-1/k}$ nearest vertices with distance $k/2-1$ for every vertex $v$.  This provides the basis for the subsequent procedures presented later on. 
Unfortunately, computing the nearest vertices of each vertex might require 
many rounds when $\Delta=\omega(\sqrt{n})$.  In particular, using Lenzen's 
routing\footnote{Lenzen's routing can be viewed as a $O(1)$-round algorithm applied when each vertex $v$ is a target and a sender of $O(n)$ messages.}\cite{lenzen2013route}, in the congested clique model, the vertices can 
learn their 
$2$-neighborhoods in $O(1)$ rounds, when the maximum degree is bounded by 
$O(\sqrt{n})$. Consider a vertex $v$ that is incident to a heavy vertex $u$ (of 
degree at least $O(\sqrt{n})$). Clearly $v$ has $\Omega(n^{1/2-1/k})$ vertices 
at distance $2$, but it is not clear how $v$ can learn their identities. Although, $v$ is capable of receiving $O(n^{1/2-1/k})$ messages, the heavy neighbor $u$ might need 
to send 
$n^{1/2-1/k}$ messages to each of its neighbors, thus $\Omega(n^{3/2-1/k})$ 
messages in total. To avoid this, we compute the $n^{1/2-1/k}$ nearest 
vertices in a \emph{lighter} subgraph 
$G_{light}$ of $G$ with maximum degree $\sqrt{n}$. The neighbors of heavy vertices 
might not learn their $2$-neighborhood and would be handled slightly differently in \Cref{dense}.
\begin{definition}
A vertex $v$ is \emph{heavy} if $\deg(v,G)\geq \sqrt{n}$, the set of heavy vertices is denoted by $V_{heavy}$. Let $G_{light}=G[V\setminus V_{heavy}]$. 
\end{definition}

\begin{definition}
For each vertex $u \in V(G_{light})$ define $N_{k/2-1}(u)$ to be the set of $y(u)=\min\{n^{1/2-1/k}, |\Gamma_{k/2-1}(u,G_{light})|\}$ 
closest vertices at distance at 
most $(k/2-1)$ from $u$ (breaking ties based on IDs) in $G_{light}$. Define 
$T_{k/2-1}(u)$ to be the truncated BFS tree rooted at $u$ consisting of the 
$u-v$ shortest path in $G_{light}$, for every $v \in N_{k/2-1}(u)$. 
\end{definition}

\begin{lemma}\label{lem:topology}
There exists a deterministic algorithm $\NearestNeighbors$ that within $O(\log k)$ rounds, computes the truncated BFS tree $T_{k/2-1}(u)$ for each vertex $u \in V(G_{light})$. That is, after running Alg. $\NearestNeighbors$, each $u \in V(G_{light})$ knows the entire tree $T_{k/2-1}(u)$.
\end{lemma}
%
\vspace{-5pt}
\paragraph{Algorithm $\NearestNeighbors$.}
For every integer $j\geq 0$, we say that a vertex $u$ is $j$-sparse if 
$|\Gamma_{j}(u,G_{light})|\leq n^{1/2-1/k}$, otherwise we say it is $j$-dense. The algorithm starts by 
having each non-heavy vertex compute $\Gamma_2(u,G_{light})$ in $O(1)$ rounds using Lenzen's algorithm. In 
each phase $i$, vertex $u$ collects information on vertices in its 
$\gamma(i+1)$-ball in $G_{light}$, where:\\
$$\gamma(1)=2, \mbox{~and~} \gamma(i+1)=\min\{2\gamma(i)-1,k/2\}, \mbox{~for every~} i \in \{1,\ldots, \lceil \log(k/2) \rceil\}.$$
At phase $i \in \{1,\ldots, \lceil \log(k/2) \rceil\}$ the 
algorithm maintains the invariant that a vertex $u$ holds a partial BFS tree 
$\widehat{T}_i(u)$ in $G_{light}$ consisting of the vertices 
$\widehat{N}_i(u):=V(\widehat{T}_i(u))$, such that: 
\begin{description}
\item{(I1)} For an $\gamma(i)$-sparse vertex $u$, $\widehat{N}_i(u)=\Gamma_{\gamma(i)}(u)$.
\item{(I2)} For an $\gamma(i)$-dense vertex $u$, $\widehat{N}_i(u)$ consists of the closest $n^{1/2-1/k}$ vertices to $u$ in $G_{light}$.
\end{description}
Note that in order to maintain the invariant in phase $(i+1)$, it is only required that in phase $i$, the $\gamma(i)$-sparse vertices would collect the relevant information, as for the $\gamma(i)$-dense vertices, it already holds that $\widehat{N}_{i+1}(u)=\widehat{N}_i(u)$. 
In phase $i$, each vertex $v$ (regardless of being sparse or dense) sends its partial BFS tree $\widehat{T}_i(v)$ to each vertex $u$ only if (1) $u \in \widehat{N}_i(v)$ and (2) $v \in \widehat{N}_i(u)$. This condition can be easily checked in a single round, as every vertex $u$ can send a message to all the vertices in its set $\widehat{N}_i(u)$.
Let $\widehat{N}'_{i+1}(u)=\bigcup_{v \in \widehat{N}_i(u) ~\mid~ u \in \widehat{N}_i(v)}\widehat{N}_{i}(v)$ be the 
subset of all received $\widehat{N}_i$ sets at vertex $u$.  It then uses the distances to 
$\widehat{N}_i(u)$, and the received distances to the vertices in the $\widehat{N}_i$ sets, to 
compute the shortest-path distance to each $w \in \widehat{N}_i(v)$ . As a result it computes 
the partial tree $\widehat{T}_{i+1}(u)$. The subset $\widehat{N}_{i+1}(u) \subseteq \widehat{N}'_{i+1}(u)$ 
consists of the (at most $n^{1/2-1/k}$) vertices within distance $\gamma(i+1)$ 
from $u$. This completes the description of phase $i$. We next analyze the 
algorithm and show that each phase can be implemented in $O(1)$ rounds and that 
the invariant on the $\widehat{T}_i(u)$ trees is maintained. 

\paragraph{Analysis.}
We first show that phase $i$ can be implemented in $O(1)$ rounds.
Note that by definition, $|\widehat{N}_i(u)|\leq \sqrt{n}$ for every $u$, and every $i\geq 1$.
Hence, by the condition of phase $i$, each vertex sends $O(n)$ messages and 
receives $O(n)$ messages, which can be done in $O(1)$ rounds, using Lenzen's 
routing algorithm \cite{lenzen2013route}.

We show that the invariant holds, by induction on $i$. Since all vertices first 
collected their second neighborhood, the invariant holds\footnote{This is the 
reason why we consider only $G_{light}$, as otherwise $\gamma(1)=0$ and we would not 
have any progress.} for $i=1$. Assume it holds up to the beginning of phase 
$i$, and we now show that it holds in the beginning of phase $i+1$.
If $u$ is $\gamma(i)$-dense, then $u$ should not collect any further information in phase $i$ and the assertion holds trivially. 

Consider an $\gamma(i)$-sparse vertex $u$ and let $N_{\gamma(i+1)}(u)$ be the target set of the $n^{1/2-1/k}$ closest vertices at distance $\gamma(i+1)$ from $u$.
We will fix $w \in N_{\gamma(i+1)}(u)$, and show that $w \in \widehat{N}_{i+1}(u)$ and in addition, $u$ has computed the shortest path to $w$ in $G_{light}$. Let $P$ be $u$-$w$ shortest path in $G_{light}$. If all vertices $z$ on the $\gamma(i)$-length prefix of $P$ are $\gamma(i)$-sparse, then the claim holds as $z \in \widehat{N}_i(u)$, $u \in \widehat{N}_i(z)$, and $w \in \widehat{N}_i(z')$ where $z'$ in the last vertex on the $\gamma(i)$-length prefix of $P$. Hence, by the induction assumption for the $\widehat{N}_i$ sets, $u$ can compute in phase $i$ its shortest-path to $w$. 

We next consider the remaining case where not all the vertices on the $\gamma(i)$-length path are sparse. Let $z \in \widehat{N}_i(u)$ be the first $\gamma(i)$-dense vertex (closest to $u$) on the $\gamma(i)$-length prefix of $P$.
Observe that $w \in \widehat{N}_i(z)$. Otherwise, $\widehat{N}_i(z)$ contains $n^{1/2-1/k}$ vertices that are closer to $z$ than $w$, which implies that these vertices are also closer to $u$ than $w$, and hence $w$ should not be in $N_{\gamma(i+1)}(u)$ (as it is not among the closest $n^{1/2-1/k}$ vertices to $u$), leading to contradiction. Thus, if also $u \in \widehat{N}_i(z)$, then $z$ sends to $u$ in phase $i$ its shortest-path to $w$. By the induction assumption for the $\widehat{N}_i(u),\widehat{N}_i(z)$ sets, we have that $u$ has the entire shortest-path to $w$.
It remains to consider the case where the first $\gamma(i)$-dense vertex on $P$, $z$, does not contain $u$ in its $\widehat{N}_i(z)$ set, hence it did not send its information on $w$ to $u$ in phase $i$. 
Denote $x=\dist(u,z,G_{light})$ and $y= \dist(z,w,G_{light})$, thus $x+y=|P|\leq 2\gamma(i)-1$.
Since $w \in \widehat{N}_i(z)$ but $u \notin \widehat{N}_i(z)$, we have that $y \leq x$ and $2y \leq |P|$, which implies that $y \leq \gamma(i)-1$. Let $z'$ be the vertex preceding $z$ on the $P$ path, hence $z'$ also appear on the $\gamma(i)$-length prefix of $P$ and $z' \in N_i(u)$. By definition, $z'$ is $\gamma(i)$-sparse and it also holds that $u \in \widehat{N}_i(z')$. Since $\dist(z',w,G_{light})=y+1\leq \gamma(i)$, it holds that $w \in \widehat{N}_i(z')$. Thus, $u$ can compute the $u$-$w$ shortest-path using the $z'$-$w$ shortest-path it has received from $z'$. For an illustration, see \Cref{fig:nearest}. 

\begin{figure}[t]
  \begin{minipage}[c]{0.38\textwidth}
    \caption{\label{fig:nearest} Shown is a path $P$ between $u$ and $w$ where $z$ is the first dense vertex on the $\gamma(i)$-length prefix of $P$. If $u \notin \widehat{N}_i(z)$ then $u,w \in \widehat{N}_i(z')$.}
  \end{minipage}\hfill
    \begin{minipage}[c]{0.61\textwidth}\vspace{-10mm}
\includegraphics[scale=0.45]{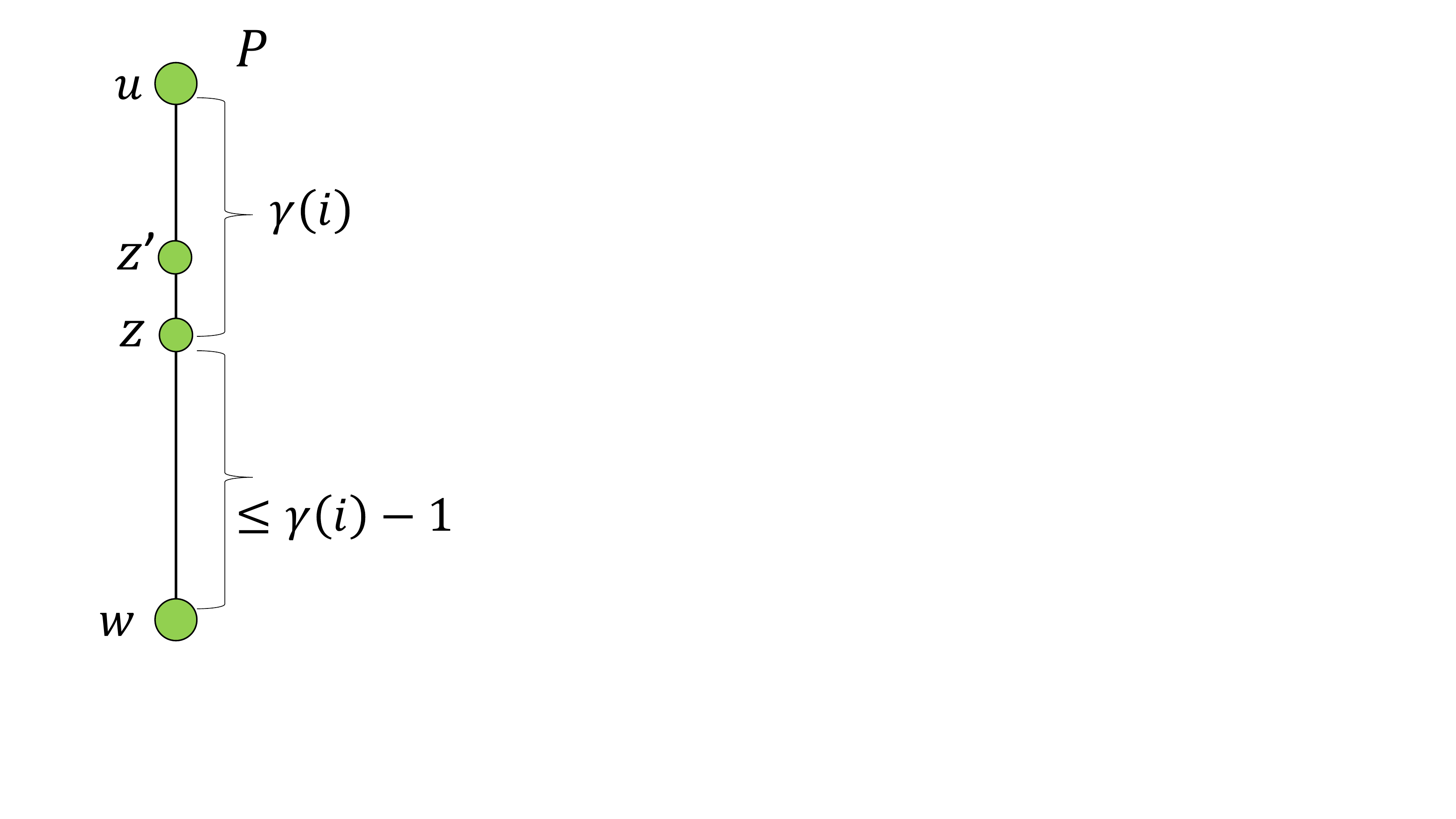}
  \end{minipage}
\vspace{-10mm}
\end{figure}

\subsection{Dividing $G$ into Sparse and Dense Regions}
Thanks for Alg. $\NearestNeighbors$ every non-heavy vertex $v$ computes the sets $N_{k/2-1}(v)$ and the corresponding tree $T_{k/2-1}(v)$. The vertices are next divided into dense vertices $V_{dense}$ and sparse vertices $V_{sparse}$. Morally, the dense vertices are those that have at least $n^{1/2-1/k}$ vertices at distance at most $k/2-1$ in $G$. Since the subsets of nearest neighbors are computed in $G_{light}$ rather than in $G$, this vertex division is more delicate. 
\begin{definition}
A vertex $v$ is \emph{dense} if either (1) it is heavy, (2) a neighbor of a heavy vertex or (3) $|\Gamma_{k/2-1}(v,G_{light})|> n^{1/2-1/k}$. Otherwise, a vertex is \emph{sparse}. Let $V_{dense}, V_{sparse}$ be the dense (resp., sparse) vertices in $V$.
%
%
\end{definition}
\begin{observation}
For $k\geq 6$, for every dense vertex $v$ it holds that $|\Gamma_{k/2-1}(v,G)|\geq n^{1/2-1/k}$.
\end{observation}
\begin{proof}
If a vertex $v$ is incident to a heavy vertex, then it has at least $\sqrt{n}$ vertices at distance $2\leq k/2-1$. 
Since $G_{light} \subseteq G$, a non-sparse vertex $v$ it holds that $|\Gamma_{k/2-1}(v,G)|\geq |\Gamma_{k/2-1}(v,G_{light})|\geq n^{1/2-1/k}$.
\end{proof}
The edges of $G$ are partitioned into:
$$
E_{dense}=\left((V_{dense} \times V_{dense})\cap E(G)\right), ~E_{sparse}=(V_{sparse} \times V_{dense})\cap E(G)$$
Since all the neighbors of heavy vertices are dense, it also holds that $E_{sparse}=(V_{sparse} \times (V\setminus V_{heavy}))\cap E(G_{light})$.%

%
\paragraph{Overview of the Spanner Constructions.}
The algorithm contains two subprocedures, the first takes care of the sparse 
edge-set by constructing a spanner $H_{sparse} \subseteq G_{sparse}$ and the 
second takes care of the dense edge-set by constructing $H_{dense}\subseteq G$. 
Specifically, these spanners will satisfy that for every $e=(u,v) \in G_i$, 
$\dist(u,v,H_{i})\leq 2k-1$ for $i \in \{sparse,dense\}$. We note that the 
spanner $H_{dense} \subseteq G$ rather than being contained in $G_{dense}$. The 
reason is that the spanner $H_{dense}$ might contain edges incident to sparse 
vertices as will be shown later. 
The computation of the spanner $H_{sparse}$ for the sparse edges, $E_{sparse}$, 
is done by letting each sparse vertex locally simulating a local spanner algorithm. The computation of $H_{dense}$ is based on applying two levels of clustering as in Baswana-Sen. The selection of the cluster centers will be made by applying an hitting-set algorithm.
%
%

%% file: sparse-full.tex
\section{Handling the Sparse Subgraph}\label{sec:detsanner}
In the section, we construct the spanner $H_{sparse}$ that will provide a bounded stretch for the sparse edges. As we will see, the topology collected by applying Alg. $\NearestNeighbors$ allows every sparse vertex to \emph{locally} simulate a deterministic spanner algorithm in its collected subgraph, and deciding which of its edges to add to the spanner based on this local view.

Recall that for every sparse vertex $v$ it holds that 
$|\Gamma_{k/2-1}(v,G_{light})|\leq n^{1/2-1/k}$ where $G_{light}=G[V \setminus 
V_{heavy}]$ and that $E_{sparse}=(V_{sparse}\times V_{dense})\cap E(G)$. Let 
$G_{sparse}(u)=G_{sparse}[\Gamma_{k/2-1}(u,G)]$.  By applying Alg. $\NearestNeighbors$, and letting sparse vertices sends their edges to the sparse vertices in their $(k/2-1)$ neighborhoods in $G_{light}$, we have:
\begin{claim}
There exists a $O(\log k)$-round deterministic algorithm, that computes for each sparse vertex $v$  its subgraph $G_{sparse}(v)$.
\end{claim}
\begin{proof}
By running Alg. $\NearestNeighbors$, every sparse vertex computes all the vertices in $\Gamma_{k/2-1}(v,G_{light})$. Note that all the neighbors of a sparse vertex are non-heavy and thus $G_{sparse} \subseteq G_{light}$.
Next, we let every sparse vertex $u$ broadcasts that it is sparse. Every sparse vertex $u$ sends its edges in $G_{sparse}$ to every sparse vertex $v \in \Gamma_{k/2-1}(u,G_{sparse})$. Since every sparse vertex sends $O(n)$ messages and receives $O(n)$ messages, this can be done in $O(1)$ many rounds using Lenzen's routing algorithm. Consider an edge $(x,y) \in G_{sparse}(u)$ for a sparse vertex $v$. By definition, both $x,y \in \Gamma_{k/2-1}(u,G_{sparse})$ and thus at least one endpoint is sparse, say $x$. By symmetry, it holds that $u \in \Gamma_{k/2-1}(x,G_{sparse})$ and thus $u$ has received all the edges incident to $x$. The claim follows.
\end{proof}
Our algorithm is based on an adaptation of the local algorithm of 
\cite{DerbelGPV08}, which is shown to satisfy the following in our context. 
\begin{lemma}\label{lem:alglocalsim}
There exists a deterministic algorithm $\LocalSpanner$ that constructs a $(k-3)$ spanner in the \local\ model, such that every sparse vertex $u$ decides about its spanner edges within $k/2-1$ rounds. In particular,
$u$ can simulate Alg. $\LocalSpanner$ locally on $G_{sparse}$ and for every edge $(u,z)$ not added to the spanner $H_{sparse}$, there is a path of length at most $(k-3)$ in $G_{sparse}(u) \cap H_{sparse}$.
\end{lemma}
A useful property 
of the algorithm\footnote{This algorithm works only for unweighted graphs and hence our deterministic algorithms are for unweighted graphs. Currently, there are no local deterministic algorithms for weighted graphs.} by Derbel et al.\ (Algorithm 1 in \cite{DerbelGPV08}) is that if a vertex $v$ did not terminate after $i$ rounds, 
then it must hold that $|\Gamma_i(v,G)|\geq n^{i/k}$. Thus in our context, every 
sparse vertex terminates after at most $k/2-1$ rounds\footnote{By definition we 
have that
$|\Gamma_{k/2-1}(u,G_{light})|\leq n^{1/2-1/k}$. Moreover, since $G_{sparse} 
\subseteq 
G_{light}$ it also holds that $|\Gamma_{k/2-1}(u,G_{sparse})|\leq 
n^{1/2-1/k}$.}. We also show that for simulating these $(k/2-1)$ rounds of Alg. 
$\LocalSpanner$ by $u$, it is sufficient for $u$ to know all the neighbors of 
its $(k/2-2)$ neighborhood in $G_{sparse}$ and these edges are
contained in $G_{sparse}(u)$. The analysis of Lemma \ref{lem:alglocalsim} is in \Cref{sec:localspanner}.

We next describe Alg. $\ConsSpannerSparseRegion$ that 
computes $H_{sparse}$. 
Every vertex $u$ computes $G_{sparse}(u)$ in $O(\log k)$ rounds and simulate Alg. $\LocalSpanner$ in that subgraph. Let $H_{sparse}(u)$ be the edges added to the spanner in the local simulation of Alg. $\LocalSpanner$ in $G_{sparse}(u)$. A sparse vertex $u$ sends to each sparse vertex $v \in \Gamma_{k/2-1}(u,G_{sparse})$, the set of all $v$-edges in $H_{sparse}(u)$. Hence, each sparse vertex sends $O(n)$ messages (at most $\sqrt{n}$-edges to each of its at most $\sqrt{n}$ vertices in $\Gamma_{k/2-1}(v,G_{sparse})$). In a symmetric manner, every vertex receives $O(n)$ messages and this step can be done in $O(1)$ rounds using Lenzen's algorithm. The final spanner is given by $H_{sparse}=\bigcup_{u \in V_{sparse}}H_{sparse}(u)$.
The stretch argument is immediate by the correctness of Alg. $\LocalSpanner$ and the fact that all the edges added to the spanner in the local simulations are indeed added to $H_{sparse}$. The size argument is also immediate since we only add edges that Alg. $\LocalSpanner$ would have added when running by the entire graph.

\begin{mdframed}[hidealllines=false]
Algorithm $\ConsSpannerSparseRegion$ (Code for a sparse vertex $u$)
\begin{enumerate}	
	\item Apply Alg. $\NearestNeighbors$ to compute $G_{sparse}(u)$ for each sparse vertex $u$.
	\item Locally simulate Alg. $\LocalSpanner$ in $G_{sparse}(u)$ and let $H_{sparse}(u)$ be the edges added to the spanner in $G_{sparse}(u)$.
	\item Send the edges of $H_{sparse}(u)$ to the corresponding sparse endpoints.
	\item Add the received edges to the spanner $H_{sparse}$.
\end{enumerate} 
\end{mdframed}

%% file: densespanner.tex
\vspace{-15pt}
\section{Handling the Dense Subgraph}\label{dense}
In this section, we present the construction of the spanner $H_{dense}$ satisfying that $\dist(u,v,H_{dense})\leq 2k-1$ for every $(u,v) \in E_{dense}$. Here we enjoy the fact the $(k/2-1)$ neighborhood of each dense vertex is large and hence there exists a small hitting that covers all these neighborhoods. 
The structure of our arguments is as follows. First, we describe a deterministic construction of $H_{dense}$ using an hitting-set algorithm as a black box. This would immediately imply a randomized spanner construction in $O(\log k)$-rounds.  Then in \Cref{sec:hitting}, we fill in this last missing piece and show deterministic constructions of hitting sets.

\paragraph{Constructing spanner for the dense subgraph via hitting sets.} 
Our goal is to cluster all dense vertices into small number of low-depth clusters. This translates into the following \emph{hitting-set} problem defined in \cite{bhattacharyya2012transitive,HittingSet,GFDetSpanner18}: 
Given a collection $\mathcal{S}=\{S(v) ~\mid~ v \in V'\}$ where each 
$|S(v)|\geq \Delta$ and $\bigcup_{ v \in V'}S(v) \subseteq V''$, compute a 
subset $Z \subseteq V''$ of cardinality $O(|V''|\log n/\Delta)$ that 
intersects (i.e., \emph{hits}) each subset $S \in \mathcal{S}$. A hitting-set of size $O(|V''|\log n/\Delta)$ is denoted as \emph{small} hitting-set.

We prove the next lemma by describing an the construction of the spanner $H_{dense}$ \emph{given} an algorithm $\cA$ that computes small hitting sets. In \Cref{sec:hitting}, we complement this lemma by describing several constructions of hitting sets.
\begin{lemma}\label{lem:spannerhit}
Let $G=(V,E)$ be an $n$-vertex graph, let $V',V'' \subset V$ and 
$\mathcal{S}=\{S_u \subset V : u \in V'\}$ be a set of subsets such that each 
node $u \in V'$ knows the set $S_u$, $|S_u| \ge \Delta$ for any $\Delta \in 
[1,n]$ and $\bigcup S_u \subseteq V''$. Let $\cA$ be a hitting set algorithm that constructs a hitting set $Z$ 
for $\mathcal{S}$ such that $|Z|=O(\log n|V''|/\Delta)$ in $r_{\cA}$ 
rounds. Then, there exists a \emph{deterministic} algorithm 
$\ConsSpannerDenseRegion$ for constructing $H_{dense}$ within $O(\log 
k+r_{\cA})$ rounds.
\end{lemma}
The next definition is useful in our context.

\paragraph{$\ell$-depth Clustering.}
A \emph{cluster} is a subset of vertices and a clustering $\mathcal{C}=\{C_1,\ldots, C_\ell\}$ consists of vertex disjoint subsets. For a positive integer $\ell$, a clustering $\mathcal{C}$ is a $\ell$-depth clustering if for each cluster $C \in \mathcal{C}$, the graph $G$ contains a tree of depth at most $\ell$ rooted at the cluster center of $C$ and spanning all its vertices.

\subsection{Description of Algorithm $\ConsSpannerDenseRegion$}\label{sec:denseg}
The algorithm is based on clustering the dense vertices in two levels of clustering, in a Baswana-Sen like manner. The first clustering $\mathcal{C}_1$ is an $(k/2-1)$-depth clustering covering all the dense vertices. The second clustering, $\mathcal{C}_2$ is an $(k/2)$-depth clustering that covers only a \emph{subset} of the dense vertices. For $k$ odd, let $\mathcal{C}_2$ be equal to $\mathcal{C}_1$.  

\paragraph{Defining the first level of clustering.}
Recall that by employing Algorithm $\NearestNeighbors$, every non-heavy vertex $v \in G_{light}$ knows the set $N_{k/2-1}(v)$ containing its $n^{1/2-1/k}$ nearest neighbors in $\Gamma_{k/2-1}(v,G_{light})$. For every heavy vertex $v$, let $N_{k/2-1}(v)=\Gamma(v,G)$. 
Let $V'$ be the set of all non-heavy vertices that are neighbors of heavy vertices. By definition, $V' \subseteq V_{dense}$. Note that for every dense vertex $v \in V_{dense}\setminus V'$, it holds that $|N_{k/2-1}(v)| \geq n^{1/2-1/k}$. The vertices $u$ of $V'$ are in $G_{light}$ and hence have computed the set $N_{k/2-1}(u)$, however, there is in guarantee on the size of these sets. 

To define the clustering of the dense vertices, Algorithm 
$\ConsSpannerDenseRegion$ applies the hitting-set algorithm $\cA$ on the 
subsets $\mathcal{S}_1=\{N_{k/2-1}(v) ~\mid~ v \in V_{dense}\setminus V'\}$. 
Since every set in $\mathcal{S}_1$ has size at least $\Delta:=n^{1/2-1/k}$, the 
output of algorithm $\cA$ is a subset $Z_1$ of cardinality $O(n^{1/2+1/k})$ that 
hits all the sets in $\mathcal{S}_1$.

We will now construct the clusters in $\mathcal{C}_1$ with $Z_1$ as the cluster 
centers. 
To make sure that the clusters are vertex-disjoint and connected, we first compute the clustering in 
the subgraph $G_{light}$, and then cluster the remaining dense vertices that are not yet clustered. 
For every $v \in G_{light}$ (either dense or sparse), we say that $v$ is \emph{clustered} if 
 $Z_1 \cap N_{k/2-1}(v) \neq \emptyset$. In particular, every dense vertex $v$ for which $|\Gamma_{k/2-1}(v,G_{light})|\geq n^{1/2-1/k}$ is clustered (the neighbors of heavy vertices are either clustered or not).  For every clustered vertex $v \in G_{light}$ (i.e., even sparse ones), let $c_1(v)$, denoted hereafter the \emph{cluster center} of $v$, be the closest vertex to $v$ in $Z_1 \cap N_{k/2-1}(v)$, breaking shortest-path ties based on IDs. Since $v$ knows the entire tree $T_{k/2-1}(v)$, it knows the distance to all the vertices in $N_{k/2-1}(v)$ and in addition, it can compute its next-hop $p(v)$ on the $v$-$c_1(v)$ shortest path in $G_{light}$. Each clustered vertex $v \in G_{light}$,   
adds the edge $(v,p(v))$ to the spanner $H_{dense}$. It is easy to see that this defines a $(k/2-1)$-depth clustering in $G_{light}$ that covers all dense vertices in $G_{light}$. In particular, each cluster $C$ has in the spanner a tree of depth at most $(k/2-1)$ that spans all the vertices in $C$. Note that in order for the clusters $C$ to be connected in $H_{dense}$, it was crucial that all vertices in $G_{light}$ compute their cluster centers in $N_{k/2-1}(v)$, if such exists, and not only the dense vertices.
We next turn to cluster the remaining dense vertices.
For every heavy vertex $v$, let $c_1(v)$ be its closest vertex in $\Gamma(v,G) \cap Z_1$. It then adds the edge $(v,c_1(v))$ to the spanner $H_{dense}$ and broadcasts its cluster center $c_1(v)$ to all its neighbors. Every neighbor $u$ of a heavy vertex $v$ that is not yet clustered, joins the cluster of $c_1(v)$ and adds the edge $(u,v)$ to the spanner. Overall, the clusters of $\mathcal{C}_1$ centered at the subset $Z_1$ cover all the dense vertices. In addition, all the vertices in a cluster $C$ are connected in $H_{dense}$ by a tree of depth $k/2-1$. 
Formally, $\mathcal{C}_1=\{C_1(s), ~\mid~ s \in Z_1\}$ where
$C_1(s)=\{v ~\mid~ c_1(v)=s\}$. 

\paragraph{Defining the second level of clustering.}
Every vertex $v$ that is clustered in $\mathcal{C}_1$ broadcasts its cluster center $c_1(v)$ to all its neighbors. 
This allows every dense vertex $v$ to compute the subset $N_{k/2}(v)=\{ s \in Z_1 ~\mid~ E(v,C_1(s),G)\neq \emptyset\}$ consisting of the centers of its adjacent clusters in $\mathcal{C}_1$.
Consider two cases depending on the cardinality of $N_{k/2}(v)$. Every vertex $v$ with $|N_{k/2}(v)|\leq n^{1/k}\log n$, adds to the spanner $H_{dense}$ an arbitrary edge in $E(v,C_1(s),G)$ for every $s \in N_{k/2}(v)$. 
It remains to handle the remaining vertices $V'_{dense}=\{ v \in V_{dense} ~\mid~ |N_{k/2}(v)|> n^{1/k}\log n\}$.
 These vertices would be clustered in the second level of clustering $\mathcal{C}_2$. 
To compute the centers of the clusters in $\mathcal{C}_2$, the algorithm applies the hitting-set algorithm 
$\cA$ on the collection of subsets $\mathcal{S}_2=\{N_{k/2}(v) ~\mid~ v \in V'_{dense}\}$ with $\Delta=n^{1/k}\log n$ and $V''=Z_1$. The output of $\cA$ is a subset $Z_2$ of cardinality $O(|Z_1|\log n/\Delta)=O(\sqrt{n}\log n)$ that hits all the subsets in $\mathcal{S}_2$.
The $2^{nd}$ cluster-center $c_2(v)$ of a vertex $v \in V'_{dense}$ is chosen to be an arbitrary $s \in N_{k/2}(v)\cap Z_2$. The vertex $v$ then adds some edge $(v,u) \in E(v,C_1(s),G)$ to the spanner $H_{dense}$. Hence, the trees spanning rooted at $s \in Z_2$ are now extended by one additional layer resulting in a $(k/2)$-depth clustering.  

\paragraph{Connecting adjacent clusters.}
Finally, the algorithm adds to the spanner $H_{dense}$ a single edge between each pairs of adjacent clusters $C,C' \in \mathcal{C}_1 \times \mathcal{C}_2$,  this can be done in $O(1)$ rounds as follows. 
Each vertex broadcasts its cluster ID in $\mathcal{C}_2$. Every vertex $v \in C$ for every cluster $C \in \mathcal{C}_1$ picks one incident edge to each cluster $C' \in \mathcal{C}_2$ (if such exists) and sends this edge to the corresponding center of the cluster of $C'$ in $\mathcal{C}_2$. Since a vertex sends at most one message for each cluster center in $\mathcal{C}_2$, this can be done in $O(1)$ rounds. Each cluster center $r$ of the cluster $C'$ in $\mathcal{C}_2$ picks one representative edge among the edges it has received for each cluster $C \in \mathcal{C}_1$ and sends  a notification about the selected edge to the endpoint of the edge in $C$. Since the cluster center sends at most one edge for every vertex this take one round. Finally, the vertices in the clusters $C \in \mathcal{C}_1$ add the notified edges (that they received from the centers of $\mathcal{C}_2$) to the spanner.
This completes the description of the algorithm. We now complete the proof of Lemma \ref{lem:spannerhit}.
\Proof
Recall that we assume $k\geq 6$ and thus $|\Gamma_{k/2-1}(v)|\geq n^{1/2-1/k}$, for every $v \in V_{dense}$. 
We first show that for every $(u,v) \in E_{dense}$, $\dist(u,v,H_{dense})\leq 2k-1$. 
The clustering $\mathcal{C}_1$ covers all the dense vertices. 
If $u$ and $v$ belong to the same cluster $C$ in $\mathcal{C}_1$, the claim follows as $H_{dense}$ contains an 
$(k/2-1)$-depth tree that spans all the vertices in $C$, thus $\dist(u,v,H_{dense})\leq k-2$. From now on assume that $c_1(u)\neq c_1(v)$. 
We first consider the case that for both of the endpoints it holds that $|N_{k/2}(v)|, |N_{k/2}(u)|\leq n^{1/k}\log n$. In such a case, since $v$ is adjacent to the cluster $C_1$ of $u$, the algorithm adds to $H_{dense}$ at least one edge in $E(v,C_1,G)$, let it be $(x,v)$. We have that $\dist(v,u,H_{dense})\leq \dist(v,x,H_{dense})+\dist(x,u,H_{dense})\leq k-1$ where the last inequality holds as $x$ and $u$ belong to the same cluster $C_1$ in $\mathcal{C}_1$.
Finally, it remains to consider the case where for at least one endpoint, say $v$, it holds that $|N_{k/2}(v)|> n^{1/k}\log n$. In such a case, $v$ is clustered in $\mathcal{C}_2$. Let $C_1$ be the cluster of $u$ in $\mathcal{C}_1$ and let $C_2$ be the cluster of $v$ in $\mathcal{C}_2$. Since $C_1$ and $C_2$ are adjacent, the algorithm adds an edge in $E(C_1,C_2,G)$, let it be $(x,y)$ where $x,u \in C_1$ and $y,v \in C_2$. We have that $\dist(u,v,H_{dense})\leq \dist(u,x,H_{dense})+\dist(x,y,H_{dense})+\dist(y,v,H_{dense})\leq 2k-1$, where the last inequality holds as $u,x$ belong to the same $(k/2-1)$-depth cluster $C_1$, and $v,y$ belong to the same $(k/2)$-depth cluster $C_2$.
Finally, we bound the size of $H_{dense}$. Since the clusters in $\mathcal{C}_1, \mathcal{C}_2$ are vertex-disjoint, the trees spanning these clusters contain $O(n)$ edges. For each unclustered vertex in $\mathcal{C}_2$, we add $O(n^{1/k}\log n)$ edges. By the properties of the hitting-set algorithm $\cA$ it holds that $|Z_1|=O(n^{1/2-1/k}\cdot \log n)$ and $|Z_2|=O(n^{1/2}\cdot \log n)$. Thus adding one edge between each pair of clusters adds $|Z_1|\cdot |Z_2|=O(n^{1+1/k}\cdot \log^2 n)$ edges. 
\QED

\textbf{Randomized spanners in $O(\log k)$ rounds.}
We now complete the proof of \Cref{lem:randspanner}.
For an edge $(u,v) \in E_{sparse}$, the correctness follows by the correctness of Alg. $\LocalSpanner$. We next consider the dense case. Let $\cA$ be the algorithm where each $v \in V'$ is added into $Z$ with probability of $\log /\Delta$. By Chernoff bound, we  get that w.h.p. $|Z|=O(|V'|\log n/\Delta)$ and $Z \cap S_i \neq \emptyset$ for every $S_i \in \mathcal{S}$. The correctness follows by applying Lemma \ref{lem:spannerhit}. 
\QED
\newpage
\begin{mdframed}[hidealllines=false]
Algorithm $\ConsSpannerDenseRegion$ 
\begin{enumerate}	
  \item Compute an $(k/2-1)$ clustering $\mathcal{C}_1=\{C(s) ~\mid~ s \in Z_1\}$ centered at subset $Z_1$.
	\item For every $v \in V_{dense}$, let $N_{k/2}(v)=\{ s \in Z_1 ~\mid~ E(v, C_1(s),G)\neq \emptyset\}$.
	\item For every $v \in V_{dense}$ with $|N_{k/2}(v)|\leq n^{1/k}\log n$, add to the spanner one edge in 
	$E(v, C(s),G)$ for every $s \in N_{k/2}(v)$.
	\item Compute an $(k/2)$ clustering $\mathcal{C}_2$ centered at $Z_2$ to cover the remaining dense vertices.
	\item Connect (in the spanner) each pair of adjacent clusters $C, C' \in \mathcal{C}_1 \times \mathcal{C}_2$ .
\end{enumerate} 
\end{mdframed}

%% file: hitting-set-full.tex
\section{Derandomization of Hitting Sets}\label{sec:hitting}
\subsection{Hitting Sets with Short Seeds}

The main technical part of the deterministic construction is to completely 
derandomize the randomized hitting-set algorithm. We show two hitting-set constructions
with different tradeoffs. The first construction is
based on pseudorandom generators (PRG) for DNF formulas. The 
PRG will have a seed of length $O(\log n (\log \log n)^3)$. This would serve the basis for the construction of \Cref{lem:detspanner}. The second hitting-set construction is based on $O(1)$-wise
independence, it uses a small seed of length $O(\log n)$ but yields 
a larger hitting-set. This would be the basis for the construction of \Cref{lem:detspanner2}. 

We begin by setting up some notation.
For a set $S$ we denote by $x \sim S$ a uniform sampling from $S$. For a function $\PRG$ and an index $i$, let $\PRG(s)_i$ 
the $\ith{i}$ bit of $\PRG(s)$. 

\begin{definition}[Pseudorandom Generators]
A generator $\PRG \colon \BB^r \to \BB^n$ is an $\epsilon$-pseudorandom
generator (PRG) for a class $\cC$ of Boolean functions if for every $f \in 
\cC$:
$$
|\Exp{x \sim \BB^n}{f(x)} - \Exp{s \sim \BB^r}{f(\PRG(s))}| \le 
\epsilon.
$$
We refer to $r$ as the seed-length of the generator and say $\PRG$ is explicit 
if there is an efficient algorithm to compute $\PRG$ that runs in time 
$poly(n, 1/\epsilon)$.
\end{definition}


\begin{theorem}\label{thm:prg}
For every $\epsilon=\epsilon(n) > 0$, there exists an explicit pseudoranom 
generator, $\PRG 
\colon 	\BB^r \to \BB^n$ that fools all read-once DNFs on $n$-variables with 
error at most $\epsilon$ and seed-length $r = O((\log(n/\epsilon)) \cdot (\log 
\log(n/\epsilon))^3)$.
\end{theorem}

Using the notation above, and \Cref{thm:prg} we formulate and prove 
the following Lemma:
\begin{lemma}\label{lem:hittingrandloglog}
Let $S$ be subset of $[n]$ where $|S| \ge 
\Delta$ for some parameter $\Delta \le n$ and let $c$ be any constant. 
Then, there 
exists a family of hash functions $\cH = \{h \colon [n] \to \BB \}$ such 
that choosing a random function from $\cH$ takes $r=O(\log n \cdot (\log\log n)^3)$ 
random bits and for $Z_h = \{u \in [n] : h(u)=0 \}$ it holds that:\\
(1) $\prob{h}{ |Z_h| \le \widetilde{O}(n/\Delta)} \ge 2/3$, and
(2) $\prob{h}{ S \cap Z_h \ne \emptyset } \ge 
	1-1/n^c$.
\end{lemma}
\begin{proof}
We first describe the construction of $\cH$. Let $p=c'\log n/\Delta$ for some large constant $c'$ (will be set later), and let 
$\ell=\lfloor \log 1/p \rfloor$. Let 
$\PRG \colon \BB^r \to \BB^{n\ell}$ be the PRG constructed in 
\Cref{thm:prg} for $r=O(\log n\ell \cdot (\log\log n\ell)^3)=O(\log n \cdot 
(\log\log n)^3)$ and for $\epsilon = 1/n^{10c}$. For a string $s$ of length $r$ 
we define 
the hash function $h_s(i)$ as follows. First, it computes $y=\PRG(s)$. 
Then, it interprets $y$ as $n$ blocks where each block 
is of length $\ell$ bits, and outputs 1 if 
and only if all the bits of the $\ith{i}$ block are 1. 
Formally, we define
$
h_s(i) = \bigwedge_{j=(i-1)\ell+1}^{i\ell}\PRG(s)_{j}.
$
We show that properties 1 and 2 hold for the set $Z_{h_s}$ where $h_s \in 
\cH$. 
We begin with property 1. For $i \in [n]$ let $X_i=h_s(i)$ be a random 
variable where $s \sim \BB^r$. Moreover, let $X=\sum_{i=1}^{n}X_i$. Using this 
notation we have that $|Z_{h_s}|=X$.
Thus, to show property 1, we need to show that
$
\Pr_{s \sim \BB^r}[X \le \widetilde{O}(n/\Delta)] \ge 2/3.
$
Let $f_i \colon \BB^{n\ell} \to \BB$ be a function that outputs 1 if the 
$\ith{i}$ 
block is all 1's. That is,
$
f_i(y)=\bigwedge_{j=(i-1)\ell+1}^{i\ell}y_{j}.
$
Since $f_i$ 
is a read-once DNF formula we have that
$$
\left |\Exp{y \sim \BB^{n\ell}}{f_i(y)} - \Exp{s \sim \BB^{r}}{f_i(\PRG(s))} 
\right| 
\le \epsilon.
$$
Therefore, it follows that
$$
\E{X} = \sum_{i=1}^{n}\E{X_i} = \sum_{i=1}^{n}\Exp{s \sim 
\BB^{r}}{f_i(\PRG(s))} \le \sum_{i=1}^{n}(\Exp{y \sim \BB^{n\ell}}{f_i(y)} + 
\epsilon) = n(2^{-\ell} + \epsilon) = \widetilde{O}(n/\Delta).$$
Then, by Markov's inequality we get that $\Pr_{s \sim \BB^r}[X > 3\E{X}] 
\le 1/3$ and thus
$$
\prob{s \sim \BB^r}{X \le \widetilde{O}(n/\Delta)} \ge
1-\prob{s \sim \BB^r}{X > 3\E{X}} \ge 2/3.
$$

We turn to show property 2. Let $S$ be any set of size at least $\Delta$ and 
let $g \colon \BB^{n\ell} \to 
\BB$ be an indicator function for the event that the set $S$ is covered. That 
is,
$$
g(y) = \bigvee_{i \in S}\bigwedge_{j=(i-1)\ell+1}^{i\ell}y_j.
$$
Since $g$ is a read-once DNF formula, and thus we have that 
$$
\left |\Exp{y \sim \BB^{n\ell}}{g(y)} - \Exp{s \sim \BB^{r}}{g(\PRG(s))} 
\right| \le \epsilon.
$$
Let $Y_i= \bigwedge_{j=(i-1)\ell+1}^{i\ell}y_j$, and let $Y=\sum_{i \in S}Y_i$. 
Then 
$
\E{Y} = \sum_{i \in S}\E{Y_i} \ge \Delta 2^{-\ell} \ge \Delta p = c'\log n.
$
Thus, by a Chernoff bound we have that $\Pr[Y = 0 ] \le \Pr[\E{Y} - Y \ge c'\log n ]  \le 1/n^{2c}$, for a large enough constant $c'$ (that depends on $c$). Together, we get that \\
$
\Pr_s[ S \cap Z_{h_s} \ne \emptyset ] = \Exp{s \sim \BB^{r}}{g(\PRG(s))} 
\ge  \Exp{y \sim \BB^{n\ell}}{g(y)} - \epsilon 
= \prob{y \sim \BB^{n\ell}}{Y \ge 1} - \epsilon 
\ge 1-1/n^{c}.
$
\end{proof}
We turn to show the second construction of dominating sets with short seed. In this construction the seed length of shorter, but the set is larger. By a direct application of Lemma 2.2 in \cite{celis2013balls}, we get the following lemma which becomes useful for showing \Cref{lem:detspanner2}.
\begin{lemma}\label{lem:hittingrandsubopt}
Let $S$ be a subset of $[n]$ where $|S| \ge 
\Delta$ for some parameter $\Delta \le n$ and let $c$ be any constant. 
Then, there 
exists a family of hash functions $\cH = \{h \colon [n] \to \BB \}$ such 
that choosing a random function from $\cH$ takes $r=O(\log n)$ 
random bits and for $Z_h = \{u \in [n] : h(u)=0 \}$ it holds that:
(1) $\prob{h}{ |Z_h| \le O(n^{17/16}/\sqrt{\Delta})} \ge 2/3$, and 
(2) $\prob{h}{ S \cap Z_h \ne \emptyset } \ge 
	1-1/n^c$.
\end{lemma}
\begin{proof}
Let $p=n^{1/16}/\sqrt{\Delta}$ and let $\cH'$ be the hash family given in 
\Cref{lem:d-wise-independent} with $d=O(c)$, $\gamma = \log n$ and $\beta=\log 
1/p$. Thus, we can sample a random hash function using $O(\log n)$ bits. Then, 
we define $h \in \cH$ using $h' \in \cH$ by $h(x)=1$ if and only if $h'(x)=0$. 
This defines $n$ random variables $X_1,\ldots,X_n$ that are $d$-wise 
independent and where $\E{X_i}=p$, Let $X=\sum_{i=1}^{n}X_i$ then 
$\E{|Z_h|}=\E{X}=np=n^{17/16}/\sqrt{\Delta}$.
By Fact \ref{fc:kwise}, we have that
$$
\pr{|Z_h| \geq 2\E{|Z_h|}} = \pr{|Z_h| \geq 2(n^{17/16}/\sqrt{\Delta})}\leq 
(1/(n^{1/8}))^{O(c)}.
$$
Fix a set $S$, and let $Y=\sum_{i \in S}X_i$. We have that $\prob{h}{ S \cap 
Z_h \ne \emptyset } = \pr{Y \ge 1}$. We know that 
$\E{Y}=p\delta=\sqrt{\Delta}n^{1/16}$.
By Fact \ref{fc:kwise}, we have that
$$
\pr{Y = 0} \le \pr{|Y - \E{Y}| \ge \E{Y} } \leq (1/(n^{1/8}))^{O(c)}.
$$
\end{proof}

\subsection{Deterministic Hitting Sets in the Congested Clique}
We next present a deterministic construction of hitting sets by means of derandomization. The round complexity of the algorithm depends on the number of random bits used by the randomized algorithms.
\begin{theorem}\label{thm:generalderand}
Let $G=(V,E)$ be an $n$-vertex graph, let $V' \subset V$, let $\mathcal{S}=\{S_u \subset V : u \in V'\}$ be a set of subsets such that each node $u \in V'$ knows the set $S_u$ and $|S_u| \ge \Delta$, and let $c$ be a constant.
Let $\cH = \{h \colon [n] \to \BB \}$ be a family of hash functions such 
that choosing a random function from $\cH$ takes $g_{\cA}(n,\Delta)$ 
random bits and for $Z_h = \{u \in [n] : h(u)=0 \}$ it holds that: 
(1) $\pr{ |Z_h| \le f_{\cA}(n,\Delta)} \ge 2/3$ and
(2) for any $u \in V'$: $\pr{ S_u \cap Z_h \ne \emptyset } \ge 1-1/n^c$. \\
Then, there exists a deterministic algorithm $\cA_{det}$ that constructs a hitting set of size $O(f_{\cA}(n,\Delta))$ in $O(g_{\cA}(n,\Delta)/\log n)$ rounds. 
\end{theorem}
\Proof
Our goal is to completely derandomize the process of finding $Z_h$ by using the 
method of conditional expectation. We follow the scheme of 
\cite{Censor-HillelPS16} to achieve this, and define two bad events that can 
occur when using a random seed of size $g=g_{\cA}(n,\Delta)$. Let $A$ be the 
event where the hitting set $Z_h$ consists of more than $f_{\cA}(n,\Delta)$ 
vertices. Let $B$ be the event that there exists an $u\in V'$ such that $S_u 
\cap Z_h=\emptyset$. Let $X_A, X_B$ be the corresponding indicator random 
variables for the events, and let $X=X_A+X_B$.

Since a random seed with $g_{\cA}(n,\Delta)$ bits avoids both of these events with high probability, we have that 
$\E{X} < 1$ where the expectation is taken over a seed of length $g$ bits. Thus, we can use the method of conditional expectations in order to get an assignment to our random coins such that no bad event occurs, \ie $X=0$. In each step of the method, we run a distributed protocol to compute the conditional expectation. Actually, we will compute a pessimistic estimator for the conditional expectation. 

Letting $X_u$ be indicator random variable for the event that $S_u$ is not hit 
by $Z_h$,
we can write our expectation as follows:
$
\E{X} = \E{X_A}+\E{X_B} = \pr{X_A=1} + \pr{X_B=1} = \pr{X_A=1} + \pr{\vee_{u} X_u=1}~$
Suppose we have a partial assignment to the seed, denoted by $Y$. Our goal is to compute the conditional expectation $\E{X|Y}$, which translates to computing $\pr{X_A=1|Y}$ and $\pr{\vee_{u} X_u=1|Y}$. Notice that computing $\pr{X_A=1|Y}$ is simple since it depends only on $Y$ (and not on the graph or the subsets $\mathcal{S}$). The difficult part is computing $\pr{\vee_{u} X_u=1|Y}$. Instead, we use a pessimistic estimator of $\E{X}$ which avoids this difficult computation. Specifically, we define the estimator:
$\Psi = X_A + \sum_{u \in V'} X_u.$
Recall that for any $u \in V'$ for a random $g$-bit length seed, it holds that $\pr{X_u=1}\leq 1/n^c$ and thus by applying a union bound over all $n$ sets, it also holds that
$\E{\Psi} = \pr{X_A=1} + \sum_u \pr{X_u=1} < 1$.
We describe how to compute the desired seed using the method of conditional expectation. We will reveal the assignment of the seed in chunks of $\ell=\lfloor \log n \rfloor$ bits. In particular, we show how to compute the assignment of $\ell$ bits in the seed in $O(1)$ rounds. Since the seed has $g$ many bits, this will yield an $O(g/\log n)$ round algorithm.


Consider the $\ith{i}$ chunk of the seed $Y_i=(y_1,\ldots, y_\ell)$ and assume that the assignment for the first $i-1$ chunks $Y_1\ldots,Y_{i-1}$ have been computed. For each of the $n$ possible assignments to $Y_i$, we assign a node $v$ that receives the conditional probability values $\pr{X_u=1|Y_1,\ldots,Y_{i}}$ from all nodes $u \in V'$.
Notice that a node $u$ can compute the conditional probability values $\pr{X_u=1|Y_1,\ldots,Y_{i}}$, since $u$ knows the IDs of the vertices in $S_u$ and thus has all the information for this computation.
The node $v$ then sums up all these values and sends them to a global leader $w$. The leader $w$ can easily compute the conditional probability $\pr{X_A=1|Y}$, and thus using the values it received from all the nodes it can compute $\E{X|Y}$ for of the possible $n$ assignments to $Y_i$. Finally, $w$ selects the assignment $(y^*_1,\ldots,y^*_\ell)$ that minimizes the pessimistic estimator $\Psi$ and broadcasts it to all nodes in the graph. After $O(g/\log n)$ rounds $Y$ has been completely fixed such that $X<1$. Since $X_A$ and $X_B$ get binary values, it must be the case that $X_A=X_B=0$, and a hitting set has been found.
\QED
Combining Lemma \ref{lem:hittingrandloglog} and Lemma \ref{lem:hittingrandsubopt} with \Cref{thm:generalderand}, yields:
\begin{corollary}\label{cor:hittingsetdet}
Let $G=(V,E)$ be an $n$-vertex graph, let $V',V'' \subset V$, let $\mathcal{S}=\{S_u \subset V : u \in V'\}$ be a set of subsets such that each node $u \in V'$ knows the set $S_u$, such that $|S_u| \ge \Delta$ and $\bigcup S_u \subseteq V''$. 
Then, there exists deterministic algorithms $\cA_{det}, \cA'_{det}$ in the congested 
clique model that construct a hitting set $Z$ for $\mathcal{S}$ such that: 
(1) $|Z|=\widetilde{O}(|V''|/\Delta)$ and $\cA_{det}$ runs in $O((\log\log 
	n)^3)$ rounds. (2) $|Z|=O(|V''|^{17/16}/\sqrt{\Delta})$ and $\cA'_{det}$ runs in $O(1)$ 
	rounds. 
\end{corollary}

\textbf{Deterministic construction in $O(\log k+O((\log\log 
	n)^3))$ Rounds.} \Cref{lem:detspanner} follows by plugging Corollary \ref{cor:hittingsetdet}(1) into Lemma \ref{lem:spannerhit}.

%% file: alter-det-full.tex
\subsection{Deterministic $O(k)$-Spanners in $O(\log k)$ Rounds}\label{sec:det2}
%
Finally, we provide a proof sketch of \Cref{lem:detspanner2}. According to 
\Cref{sec:detsanner}, it remains to consider the construction of $H_{dense}$ 
for the dense edges $E_{dense}$. Recall that for every dense vertex $v$, it 
holds that $|\Gamma_{k/2}(v,G)|\geq n^{1/2-1/k}$. Similarly to 
\Cref{lem:spannerhit}, we construct a $(k/2-1)$ dominating set $Z$ for the 
dense vertices. However, to achieve the desired round complexity, we use the  
$O(1)$-round hitting set construction of Cor. \ref{cor:hittingsetdet}(II) with 
parameters of $\Delta=n^{1/2-1/k}$ and $V'=V$. The output is then a hitting set 
$Z$ of cardinality $O(n^{13/16+1/(2k)})=O(n^{7/8})$ that hits all the $(k/2-1)$ 
neighborhoods of the dense vertices. Then, as in Alg. 
$\ConsSpannerDenseRegion$, we compute a $(k/2-1)$-depth clustering 
$\mathcal{C}_1$ centered at $Z$. The key difference to Alg. 
$\ConsSpannerDenseRegion$ is that $|Z|$ is too large for allowing us to add an 
edge between each pair of adjacent clusters, as this would result in a spanner 
of size $O(|Z|^2)$. Instead, we essentially contract the clusters of 
$\mathcal{C}_1$ (i.e., contracting the intra-cluster edges) and construct the 
spanner recursively in the contracted graph $G''$. Every contracted 
node in $G''$ corresponds to a cluster with a small strong diameter in the 
spanner. 
Specifically, $G''$ will be decomposed into sparse and dense regions (as in our 
previous constructions). Handling the sparse part is done 
deterministically by applying Alg.  $\ConsSpannerSparseRegion$. To handle the 
dense case, we apply the hitting-set algorithm of 
Cor. \ref{cor:hittingsetdet}(II) to cluster the dense nodes (which are in fact, 
contracted nodes) into $|V(G'')|/\sqrt{\Delta}$ clusters for 
$\Delta=n^{1/2-1/k}$. After $O(1)$ repetitions of the above, we will be left 
with a contracted graph with $o(\sqrt{n})$ vertices. At this point, we will 
connect each pair of clusters (corresponding to these contracted nodes)  in the 
spanner. A na\"ive implementation of such an approach would yield a spanner 
with stretch $k^{O(1)}$, as the diameter of the clusters induced by the 
contracted nodes is increased by a $k$-factor in each of the phases. To avoid 
this blow-up in the stretch, we enjoy the fact that already after the first 
phase, the contracted graph $G'$ has $O(n^{13/16+o(1)})$ nodes and hence we can 
allow our-self to compute a $(2k'-1)$ spanner for $G'$ with $k'=7$ as this 
would add $O(n)$ edges to the final spanner. Since in each of the phases -- 
except for the first-- the stretch parameter is \emph{constant}, the stretch 
would be bounded by $O(k)$, and the number of edges by $O(n^{1+1/k})$. 
%
%
%

\paragraph{A detailed description of the algorithm.}
The algorithm consists of $O(1)$ phases. In each phase $i\geq 0$, we are given a virtual graph $\mathcal{G}_{i}=(\mathcal{V}_i,\mathcal{E}_i)$ where $\mathcal{V}_0=V$ and $\mathcal{V}_{i\geq 1}$ is the computed hitting set of phase $(i-1)$.
We are also given a clustering $\mathcal{C}_i$ centered at the vertices $\mathcal{V}_i$. For every $v \in  \mathcal{V}_i$, let $C_{i}(v)\subseteq V$ be the cluster of $v$ in $\mathcal{C}_i$. Initially, we have $\mathcal{G}_0=G$ and $C_0(v)=\{v\}$ for every $v$. We keep the following invariant for each edge $(u,v) \in \mathcal{E}_i$: (1) it corresponds to a unique $G$-edge between the clusters $C_{i}(u)$ and $C_{i}(v)$, and (2) both $u$ and $v$ know the endpoints of this $G$-edge. Note that eventhough the edges of $\mathcal{G}_{i}$ are virtual, by property (2), each vertex $v \in \mathcal{V}_i$ knows its edges in $\mathcal{E}_i$, and hence we can employ any graph algorithm on $\mathcal{G}_{i}$ at the same round complexity as if the edges were in $G$. We next describe phase $i$ that constructs a subgraph $\mathcal{H}_{i} \subseteq \mathcal{G}_{i}$. At the end of that phase, the vertices will add to the spanner $H$, the $G$-edges corresponding to the virtual edges in $\mathcal{H}_{i}$. 

Given the virtual graph $\mathcal{G}_i$ in phase $i$ we do as follows. Let $\mathcal{V}_{i,heavy} \subseteq \mathcal{V}_i$ be the vertices with degree at least $\sqrt{n}$ in $\mathcal{G}_i$. Let $\mathcal{G}_{i,light}=\mathcal{G}_i[\mathcal{V}_{i} \setminus \mathcal{V}_{i,heavy}]$ be the induced subgraph on the non-heavy vertices.  First, we apply the $O(\log (k'))$-round algorithm $\NearestNeighbors$ on $\mathcal{G}_{i,light}$ with parameter $k'$, where
$k'=k$ for $i=0$, and $k'=7$ for $i\geq 1$. We define $\mathcal{V}_{i,dense}$ and $\mathcal{V}_{i,sparse}$ in the exact same manner as in our previous construction (only with using $k'$ as the stretch parameter). 
The graph $\mathcal{G}_i$ is partitioned into sparse edges $\mathcal{E}_{i,sparse}$ and dense edges $\mathcal{E}_{i,dense}$ also as in the previous sections. 

To handle the sparse subgraph $\mathcal{G}_{i,sparse}$, we apply Alg. $\ConsSpannerSparseRegion$ with stretch parameter $k'$, resulting in the spanner $\mathcal{H}_{i,sparse} \subseteq \mathcal{G}_{i}$ for the sparse edges.
%
We next handle the dense subgraph $\mathcal{G}_{i,dense}$. The algorithm will be very similar to Alg. $\ConsSpannerDenseRegion$ of \Cref{sec:denseg}, the main difference will be that we will use the deterministic hitting-set algorithm of Cor. \ref{cor:hittingsetdet}(II) that results in a larger number of clusters. Then, 
instead of connecting each pair of clusters (as in Alg. $\ConsSpannerDenseRegion$), the clusters will be contracted into super-nodes, and the algorithm will continue recursively on that contracted graph.

For every non-heavy $v \in \mathcal{V}_i$ with $|\Gamma_{k'/2-1}(v,  \mathcal{G}_{i,light})|\geq  n^{1/2-1/k'}$, let $N_{k'/2-1}(v)$ be its closest $n^{1/2-1/k'}$ vertices in $\Gamma_{k'/2-1}(v, \mathcal{G}_{i,light})$ (as compute by Alg. $\NearestNeighbors$). For a heavy vertex $v$, let $N_{k'/2-1}(v)=\Gamma(v,\mathcal{G}_i)$. 
Let $\mathcal{V}'_i$ be the set all the non-heavy vertices that are neighbors of heavy vertices in $\mathcal{G}_{i}$. By definition, $\mathcal{V}'_i \subseteq \mathcal{V}_{i,dense}$. Note that for every dense vertex $v \in \mathcal{V}_{i,dense} \setminus \mathcal{V}'_i$, it holds that $|N_{k/2-1}(v)| \geq n^{1/2-1/k}$. The vertices $u$ of $\mathcal{V}'_i$ are in $\mathcal{G}_{i,light}$ and hence have computed the set $N_{k/2-1}(u)$, however, there is in guarantee on the size of there sets. 

We then apply the hitting-set algorithm of 
Cor. \ref{cor:hittingsetdet}(II) on the collection of sets $\mathcal{S}_i=\{N_{k/2-1}(v) ~\mid~ v \in \mathcal{V}_{i,dense} \setminus \mathcal{V}'_i\}$ with $\Delta=n^{1/2-1/k'}$, $V', V''=\mathcal{V}_i$, and compute a hitting-set 
$\mathcal{Z}_i \subseteq \mathcal{V}_i$ that hits all the sets in $\mathcal{S}_i$. The size of $\mathcal{Z}_i$ is $O(|\mathcal{V}_i|/\sqrt{\Delta})$. Next, the algorithm constructs a $(k'/2-1)$-depth clustering $\widetilde{\mathcal{C}}_i$ centered at the vertices of $\mathcal{Z}_i$ in the exact same manner as described in Alg. $\ConsSpannerDenseRegion$. This clustering is accompanied with $(k'/2-1)$-depth trees in the virtual graph $\mathcal{G}_i$ that  are added to $\mathcal{H}_{i,dense}$. 

For every $s \in \mathcal{Z}_i$, let $\widetilde{C}(s) \in \widetilde{\mathcal{C}}_i$ be the cluster of $s$ in $\mathcal{G}_{i}$ where $\widetilde{C}(s) \subseteq \mathcal{V}_i$. 
The clustering $\mathcal{C}_{i+1}$ is defined by $\mathcal{C}_{i+1}=\{C_{i+1}(s) ~\mid~ s \in \mathcal{Z}_i\}$ where $C_{i+1}(s)=\bigcup_{s' \in \widetilde{C}(s)}C_i(s') \subseteq V$ for every $s \in \mathcal{Z}_i$. 
The subgraph $\mathcal{G}_{i+1}$ is given by letting $\mathcal{V}_{i+1}=\mathcal{Z}_i$. The edge set $\mathcal{E}_{i+1}$ is defined by computing the unique $G$-edge between each pair of adjacent clusters $C_{i+1}(s)$ and $C_{i+1}(s')$ in $\mathcal{C}_{i+1}$. This can be computed in the same manner as in  Alg. $\ConsSpannerDenseRegion$. In particular, this procedure maintains the invariant that each $v \in \mathcal{V}_{i+1}$ knows the $G$-edges corresponding to its virtual edges in  $\mathcal{G}_{i+1}$. 

Let $\widetilde{H}_i=\widetilde{H}_{i,sparse}\cup \widetilde{H}_{i,dense}$ be the $\mathcal{G}_i$-edges added to the spanner in phase $i$. It remains to describe how to add the unique $G$-edge corresponding to each virtual edge in $\widetilde{H}_i$ within $O(1)$ rounds. By the invariant, for every virtual edge $(v,v') \in \widetilde{H}_i$, both endpoints $v$ and $v'$ know the $G$-edge $(x,y)$ such that $x \in C_{i}(v)$ and $y \in C_{i}(v')$. Each vertex $v \in \mathcal{V}_i$, for each of its edges $(v,v') \in  \widetilde{H}_i$ sends a message to $x \in C_{i+1}(v')$ with the identifier of the edge $(x,y)$, namely, the $G$-edge that corresponds to the virtual $(v,v')$. The vertex $x$ adds the edge $(x,y)$ to the spanner. Since the clusters $\mathcal{C}_{i}$ are vertex disjoint, each $v \in \mathcal{C}_i$ sends at most one messages to each of the vertices in the graph and hence all $G$-edges of $\widetilde{H}_i$ can be added in $O(1)$ rounds.
%
%
%

This completes the description of phase $i$. For an illustration see \Cref{fig:clusterscont}.
After $4$ phases, we show in the analysis section, that $\mathcal{G}_4$ contains $o(n^{1/2})$ vertices $\mathcal{V}_4$. We then connect each pair of clusters in $\mathcal{C}_4$ by adding one edge between each pair of adjacent clusters to the spanner. Let $H$ be the output spanner. This completes the description of the algorithm. 

\begin{figure}[t]
\includegraphics[scale=0.35]{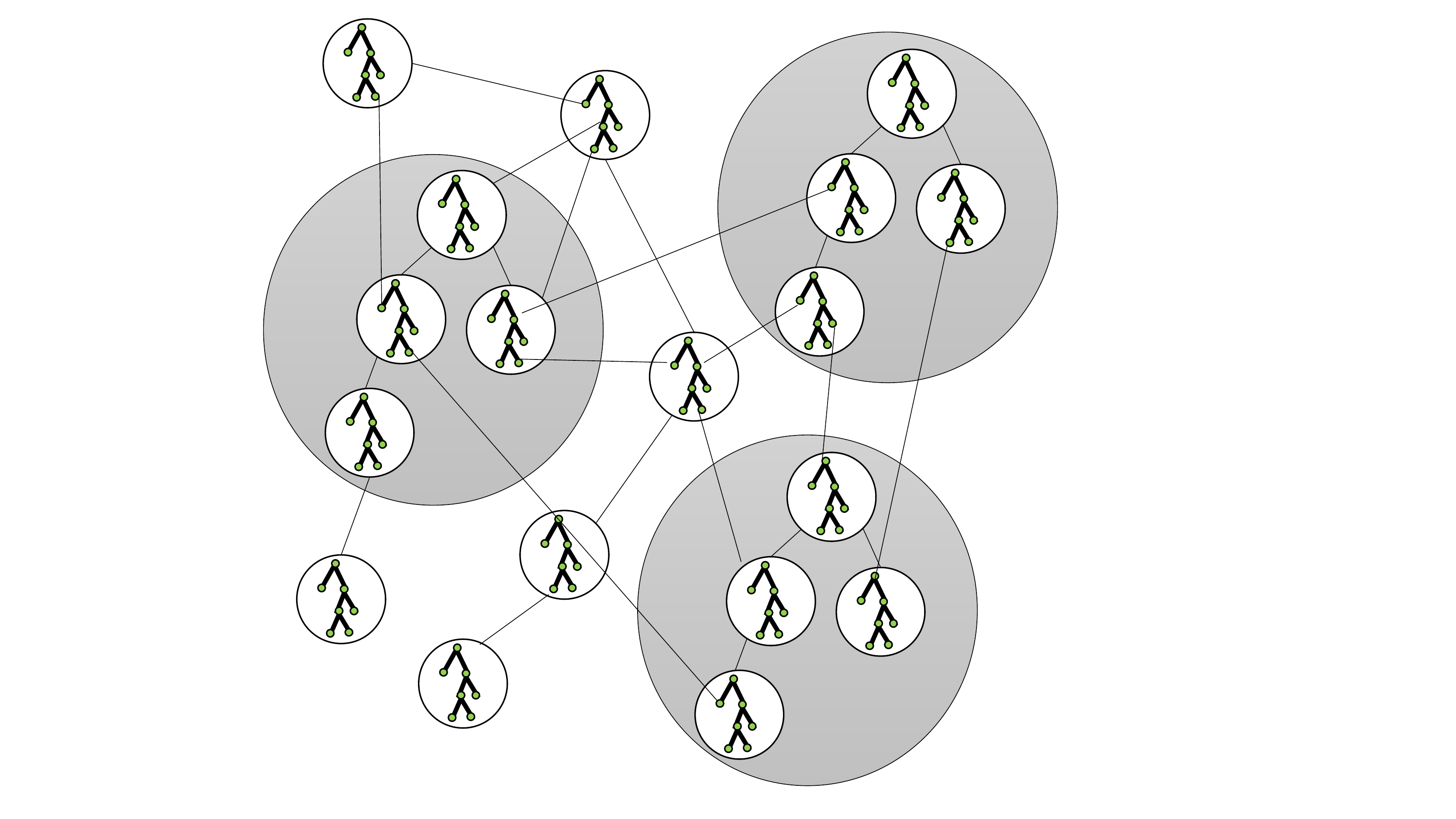}
\caption{\label{fig:clusterscont} An illustration of phase $i$. Small circles are nodes of $\mathcal{G}_i$, each corresponds to a cluster in $\mathcal{G}_{i-1}$. The greedy circles are the clusters of the dense vertices in $\mathcal{G}_i$, the centers of these clusters will be the nodes of $\mathcal{G}_{i+1}$.}
\end{figure}

\paragraph{Analysis.}

\begin{claim}\label{cl:lastphast}
$|\mathcal{V}_4|=o(\sqrt{n})$. 
\end{claim}
\begin{proof}
We show that for every $i\geq 1$, $|\mathcal{V}_i|=O(n^{1-i/4+1/2k+(i-1)/16})$. This is shown by induction on $i$. For $i=1$, $\mathcal{V}_1$ is the hitting set constructed for $V$ in $\mathcal{G}_0$. The claim holds as by applying Cor. \ref{cor:hittingsetdet}(II) with parameters $V'=V$ and $\Delta=n^{1/2-1/k}$. 
Assume that it holds up to phase $i$ and consider $i+1$. The set $\mathcal{V}_{i+1}$ is the hitting set for $\mathcal{V}_{i}$ computed by Cor. \ref{cor:hittingsetdet}(II). The claim holds as by applying Cor. \ref{cor:hittingsetdet}(II) with parameters $V'=O(n^{1-i/4+1/2k+(i-1)/16})$,$\Delta=n^{1/2-1/k'}$ and $k'=8$. We get that $|\mathcal{V}_4|=o(\sqrt{n})$ as required.
\end{proof}

\begin{claim}\label{cl:dettreesepth}
For every $C \in \mathcal{C}_i$, the spanner $H$ contains a tree of depth at most $7^{i}\cdot k$ spanning all the vertices in $C$. 
\end{claim}
\begin{proof}
The proof is by induction on $i$. For $i=0$, each node in $G$ correspond to a singleton cluster in $\mathcal{C}_0$ and since we compute a $(k/2-1)$-clustering in $G$, this is immediate. Assume that the claim holds up to phase $(i-1)$ and consider phase $i$. 
Recall that the vertices of $\mathcal{V}_i$ are the cluster centers of the clustering $\mathcal{C}_{i-1}$.
By induction assumption for $i-1$, each node $v \in \mathcal{V}_i$ is a root of tree of depth at most $x_i=7^{i-1}\cdot k$ that spans all the vertices in $C_{i-1}(r)$ in the spanner. In phase $i$, the algorithm computes a virtual tree of depth $(k'/2-1)=3$ connecting the vertices of $\mathcal{V}_{i}$. Since there is an $H$-path of length $2x_i+1$ between each two neighbors $v_1,v_2$ in $\mathcal{G}_i$, overall we get that $H$ as a tree of depth at most $3(2x_i+1)\leq 7x_i=7^{i}\cdot k$ in $H$ that spans the vertices in each cluster $C \in \mathcal{C}_i$, as required.
\end{proof}

We are now ready to complete the proof of \Cref{lem:detspanner2}. 
\Proof
We begin with stretch analysis and show that $H$ is an $O(k)$-spanner. 
Consider an edge $(u,v)$. We say that a vertex $w$ is $i$-clustered if its belongs to one of the clusters of $\mathcal{C}_i$. Without loss of generality, assume that $u$ becomes unclustered \emph{not after} $v$. Let $i$ be the maximum integer such that $u$ is $i$-clustered. If $i=4$, then both $u$ and $v$ belongs to clusters in $\mathcal{C}_4$. By \Cref{cl:dettreesepth}, the spanner $H$ contains a depth-$O(k)$ tree connecting $v,u$ to their cluster centers. Since the algorithm connects each pair of neighboring clusters, either $(u,v)$ is in $H$ or an edge $(w,z)$ was added such that $u,w$ and $v,z$ belong to the same clusters respectively. 

It remains to consider the case where $i\leq 3$. Let $s_u,s_v \in \mathcal{V}_i$ be the cluster centers of $u,v$ in $\mathcal{C}_i$ respectively.  
Since $u$ is unclustered in $\mathcal{C}_{i+1}$, this implies that $s_u$ is a sparse node and hence $s_v$ is a non-heavy node in $\mathcal{G}_i$. Hence, the virtual edge $(s_u,s_v)$ belongs to $\mathcal{G}_{i,sparse}$. By applying Alg. $\ConsSpannerSparseRegion$ to $\mathcal{G}_{i,sparse}$, we get an $15$-spanner $\mathcal{H}_{i,sparse}$. Hence, the virtual spanner $\mathcal{H}_{i,sparse}$ contains an $s_u$-$s_v$ path $\mathcal{P}=[s_u=s_1,s_2,\ldots, s_\ell=s_v]$ of length at most $15$ in $\mathcal{G}_i$.
Since the algorithm adds a $G$-edges between each neighboring clusters $C_i(s_j)$ and $C_i(s_{j+1})$ for every $j \in \{1,\ldots, \ell-1\}$ and by \Cref{cl:dettreesepth}, each $C_i(s)$ contains a depth-$(7^{i}\cdot k)$ tree in $H$ rooted at $r$ that spans all the vertices in $C_i(s)$, we get that $H$ contains an $u$-$v$ path of length $O(k)$. 

We now bound the number of edges in $H$. In each phase $i \in \{0,\ldots,4\}$, we only add edges in the construction of the spanner for the sparse regions. Phase $i=0$, computes a $(2k-1)$ spanner for the sparse edges of $G$ and hence this add $O(k \cdot n^{1+1/k})$ edges. Phase $i \in \{1,\ldots, 4\}$ computes an $15$-spanner for a virtual graph with at most $n_i=O(n^{7/8})$ vertices. The virtual spanner $\mathcal{H}_{i,sparse}$ contains $O((n^{7/8})^{1+1/8})=O(n)$ edges. As the algorithm adds an $G$-edge for each virtual edge of  $\mathcal{H}_{i,sparse}$, at most $O(n)$ edges are added to the spanner. Finally, by \Cref{cl:lastphast}, in the last phase the virtual graph contains $O(\sqrt{n})$ nodes, since the algorithm adds to the spanner a $G$-edge between the clusters of each neighboring nodes, this adds $O(n)$ edges in total. 
\QED

%% file: appendix.tex
\section{Limited Independence}
\begin{definition}[{\cite[Definition 3.31]{Vadhan12}}]
	\label{def: d-wise independent}
	For	$N,M,d \in \mathbb{N} $ such that $d \leq N$, a family of functions $\Hcal = \set{h : [N] \rightarrow
		[M]}$ is $d$-wise independent if for all distinct $x_1,x_2,...,x_d \in [N],$ the
	random variables $H(x_1),...,H(x_d)$ are independent and uniformly distributed
	in $[M]$ when $H$ is chosen randomly from $\Hcal$.
\end{definition}
In \cite{Vadhan12} an explicit construction of $\Hcal$ is presented, with parameters as stated in the next Lemma.
\begin{lemma}[{\cite[Corollary 3.34]{Vadhan12}}]
	\label{lem:d-wise-independent}
	For every $\gamma,\beta,d \in \mathbb{N},$ there is a family of $d$-wise independent functions $\mathcal{H}_{\gamma,\beta} = \set{h : \set{0,1}^\gamma \rightarrow \set{0,1}^\beta}$ such that choosing a random function from $\mathcal{H}_{\gamma,\beta}$ takes $d \cdot \max \set{\gamma,\beta}$ random bits, and evaluating
	a function from $\mathcal{H}_{\gamma,\beta}$ takes time $poly(\gamma,\beta,d)$.
\end{lemma}

\begin{fact}\label{fc:kwise}\cite{celis2013balls}
Let $X_1,\ldots, X_n \in \{0,1\}$ be $2k$-wise $\delta$-dependent random variables, for some $k \in \mathbb{N}$ and $0 \leq \delta <1$, and let $X=\sum_{i=1}^n X_i$ and $\mu=\E{X}$. Then, for any $t>0$ it holds that:
$$\pr{|X-\mu|>t}\leq 2(2nk/t^2)^k+\delta(n/t)^{2k}~.$$
\end{fact}
%
%

\section{Deterministic $(2k-1)$ Spanners for $k \in \{2,\ldots, 5\}$} \label{sec:detspannersmallk}
We will simulate the Baswana-Sen algorithm as in \cite{Censor-HillelPS16}. In particular, the algorithm constructs $k$ levels of clustering $\mathcal{C}_0=\{\{v\}, v \in V\}$, $\mathcal{C}_1, \ldots, \mathcal{C}_{k-1}$ where the $i^{th}$ clustering $\mathcal{C}_i$ consists of $O(n^{1-i/k})$ clusters centered at vertices $Z_i$. In the spanner, for every cluster $C \in \mathcal{C}_i$ with a cluster center $r \in Z_i$, there is a depth-$i$ tree spanning the vertices of $C$. To pick the centers $Z_i$ for every $i \in \{1,\ldots, k-1\}$, we apply the deterministic hitting-set algorithm. A vertex $v$ is $i$-clustered if it belongs to some of the clusters of $\mathcal{C}_i$. 

We now describe phase $i$ in details where given $\mathcal{C}_{i-1}$, we construct the clustering $\mathcal{C}_i$ and add edges to the spanner incident to the vertices that are no longer clustered in $\mathcal{C}_i$.
For every $(i-1)$-clustered vertex $v$ that is incident to less than $O(n^{1/k}\log n)$ clusters in $\mathcal{C}_{i-1}$, we add one edge connecting $v$ to each of these clusters. Let $V_i$ be the remaining $(i-1)$-clustered vertices. For each $v \in V_i$, let $S_i(v) \subseteq Z_{i-1}$ be the centers of the clusters in $\mathcal{C}_{i-1}$ that are incident to $v$. By definition, $|S_i(v)|\geq O(n^{1/k}\log n)$. We apply the deterministic algorithm 
of \Cref{cor:hittingsetdet} and compute a hitting set $Z_{i} \subseteq Z_{i-1}$ of cardinality $O(|Z_{i-1}|\cdot n^{-1/k})$. We then define an $i$-clustering $\mathcal{C}_i$ by letting each vertex $v \in V_i$ connect to one of its neighbors $w$ that belong to a cluster in $\mathcal{C}_{i-1}$ centered at $r \in Z_i$. This defines a depth-$i$ trees centered at the vertices of $Z_i$. After $k-1$ rounds, we define a clustering $\mathcal{C}_{k-1}$, at this point, every vertex adds one edge to each of its incident clusters in $\mathcal{C}_{k-1}$. 

\section{Deterministic Spanners for Sparse-Subgraphs in the \local\ Model}\label{sec:localspanner}

\begin{figure}[!h]
\begin{boxedminipage}{\textwidth}
	\vspace{3mm} \textbf{Algorithm $
	\LocalSpanner(G)$} (slightly modified version of \cite{DerbelGPV08})
	\begin{enumerate}
	 	\item $W(u)=\Gamma(u,G)$, $L,C,R(u)=\{u\}$, $\sigma=n^{1/k}$. 
		\item For $i=1$ to $k$ do:
		\begin{enumerate}
		\item Node $u$ sends $R(u)$ to its active neighbors, and receives $R(w)$ from active neighbors. 
		\item Node $u$ removes from $W$ all non-active neighbors.
		\item While $\exists w \in W$ such that $R(u)\cap R(w)\neq \emptyset$ do:
		\begin{enumerate}
		\item $W \gets W \setminus \{w\}$
		\end{enumerate}
		\item While $\exists w \in W$ and $|L|\leq i \cdot \sigma$ do:
		\begin{enumerate}
		\item $W\gets W \setminus \{v \in W ~\mid~ R(v)\cap R(w)\neq \emptyset\}$.
		\item $L \gets L \cup \{w\}$
		\item $C \gets C \cup R(w)$.
		\item If $W=\emptyset$, inactivate $u$.
		\end{enumerate}
		\item $R(u)\gets C$
		\end{enumerate}
	\end{enumerate}
\end{boxedminipage}
\caption{Deterministic construction of $(2k-1)$ spanner in the \local\ model.}
\label{fig:greedy}
\end{figure}
Let $R_i(u)$ be the set $R(u)$ at the end of round $i\geq 1$ where $R_0(u)=\{u\}$ and let $H_i$ be the spanner at the end of round $i$ where $H_0=\emptyset$. 
\begin{claim}\label{cl:spasepath}
Let $H_i$ be the edges added to the spanner at the end of round $i$. We have: \\
(1) $\dist(u,v,H_i)\leq 2i-1$ for every vertex $u$ that omitted the edge $(u,v)$ from $W$ in round $i$.\\
(2) $\dist(u,w,H_i)\leq i$ for every node $u$ that is active in round $i$ and $w \in R_i(u)$.
\end{claim}
\begin{proof}
The proof is by induction on $i$. For the induction base $i=1$, all $R(u)$ are distinct and hence no edges are omitted from $W$ and (1) holds .
Since $R_1(u)\subseteq \Gamma(u)$, (2) holds as well
 
Assume that the claims holds up to round $(i-1)$ and consider round $i$. In round $i$, an active vertex $u$ collects the $R(v)$ sets only of its active neighbors. 
We start with (1). For every edge $(u,a)$ that is omitted from the set $W$ of $u$, in round $i$, it holds that $u$ added an edge $(u, b)$ to $H_i$ in round $i$, such that there exists 
$x \in R_{i-1}(a) \cap R_{i-1}(b)$. 
By induction assumption (II), it holds that $\dist(a,x,H_{i-1}),\dist(b,x,H_{i-1})\leq (i-1)$. Hence 
$$\dist(u,a,H_i)\leq \dist(u,b,H_i)+\dist(b,x,H_i)+\dist(x,a,H_i)\leq 2(i-1)+1=2i-1~.$$
We  proceed with (2).
%
%
All the vertices $x$ added to $R_i(u)$ belong to some $R_{i-1}(w)$ where the edge $(u,w)$ was added to the spanner in round $i$. By induction assumption for $(i-1)$, $\dist(w,x,H_{i-1})\leq i-1$ and thus $\dist(u,x,H_i)\leq i$ for every $x \in R_i(w)$. 
\end{proof}

\begin{claim}\label{cl:inactivefast}
Every sparse vertex becomes inactive after $(k/2-1)$ rounds.
\end{claim}
\begin{proof}
We prove by induction on $i\geq 1$, that if $u$ is still active in round $i+1$, then $|R_i(u)|\geq n^{i/k}$ and that $R_i(u)\subseteq \Gamma_i(u,G)$. This would imply that a sparse vertex $u$ becomes inactive at the end of round $k/2-1$.

The base of the induction $i=1$ follows as initially $R_0(w)=\{w\}$ and hence there is no overlap between the sets of the neighbors, and $u$ adds $n^{1/k}$ neighbors into $R_1(u)$ hence $R_1(u) \subseteq \Gamma_1(u,G)$. 
(If $u$ adds less than $n^{1/k}$ neighbors, then it becomes inactive and we are done).
Assume that the claim holds at the end of round $i-1$ and consider round $i\leq k/2-1$.
Since $u$ is active at the end of round $(i-1)$, it implies that we added $n^{1/k}$ \emph{active} neighbors $L'$ of $u$ into $L(u)$. By applying the induction assumption for $i-1$ on each of these active neighbors, we get that $R_{i-1}(w) \subseteq \Gamma_{i-1}(w)$ and $|R_{i-1}(w)|\geq n^{(i-1)/k}$. Since the $R_{i-1}$ sets of each $w \in L'$ are vertex disjoint, and since $L' \subseteq \Gamma(u)$, the claim follows.
\end{proof}

\begin{claim}\label{cl:localsim}
If a sparse vertex $u$ knows all the edges in $G_{sparse}(u)$ then it can 
locally simulate Alg. $\LocalSpanner$. 
\end{claim}
\begin{proof}
We show that if a sparse vertex $u$ knows all the edges incident to $\Gamma_{k/2-2}(u,G_{sparse})$ then it can 
locally simulate Alg. $\LocalSpanner$. Since this edge set is contained in $G_{sparse}(u)$, this would prove the claim. 

By \Cref{cl:inactivefast}, every sparse vertex $u$ becomes inactive at the end of round $i \leq k/2-1$ of Alg. $\LocalSpanner$.
We prove by induction on $i$ that to simulate the first $i$ rounds of Alg. $\LocalSpanner$, it is sufficient for each vertex $u$ (either sparse or dense) to know all the edges incident to the vertices in $\Gamma_{i-1}(u,G_{sparse})$. For ease of notation, let $\Gamma_i(v,G_{sparse})=\Gamma_i(v)$. 
For $i=1$, since the $R_0$ sets are singletons, $u$ only needs to know its neighbors and the claim holds. 
Assume that the claim holds up to round $i-1$ and consider round $i$. 
In round $i$, $u$ should simulate the first $i-1$ rounds for each of its neighbors. By induction assumption, for simulating the first $i-1$ rounds for $w$, it is sufficient to know all edges incident to the vertices in $\Gamma_{i-1}(w)$. Hence, it is sufficient for $u$ to know all the edges incident to the vertices in $\Gamma_i(u)$ as $\Gamma_{i-1}(w) \subseteq \Gamma_i(u)$ for every $w \in \Gamma(u,G_{sparse})$. 
\end{proof}

\begin{claim}\label{cl:localsimpath}
Consider a sparse vertex $u$. For every edge $(u,v)$ not added to the spanner, when $u$ locally simulates $\LocalSpanner$ in $G_{sparse}(u)$, it holds that there is a path of length $k-3$ fully contained in $G_{sparse}(u)\cap H_{k/2-1}$ where $H_i$ is the current spanner of Alg. $\LocalSpanner$ at the end of round $i$ for every $i \in \{1,\ldots,k\}$. 
\end{claim}
\begin{proof}
Let $i$ be the round in which the edge $(u,v)$ is discarded from the set $W$ (and hence not added to $H_i$). 
For every edge $(u,v)$ that is omitted from $W$ in round $i$, there is another edge $(u,z)$ that is added to the spanner in round $i$ such that there exists $x \in R_{i-1}(v) \cap R_{i-1}(z)$. In addition, there are paths of length at most $(i-1)$ between $v-x$ and $z-x$ (by \Cref{cl:spasepath}(II)). Thus, $u$ can see two $u$-$x$ paths of length at most $i$. Since $i \leq k/2-1$ and $u$ sees all the edges incident to the vertices in $\Gamma_{k/2-2}(u,G_{sparse})$, it can see the alternative $(k-3)$-length $u$-$v$ path in $G_{sparse}(u)\cap H_{k/2-1}$. 
\end{proof}
The proof of \Cref{lem:alglocalsim} follows by \Cref{cl:localsim},\Cref{cl:localsimpath}.
